\documentclass[letterpaper,twocolumn,10pt]{article}
\usepackage{usenix}

\usepackage{tikz}
\usepackage{amsmath}

\usepackage{cite}
\usepackage{amsmath,amssymb,amsfonts}
\usepackage{graphicx}
\usepackage{textcomp}
\usepackage{xcolor}
\def\BibTeX{{\rm B\kern-.05em{\sc i\kern-.025em b}\kern-.08em
    T\kern-.1667em\lower.7ex\hbox{E}\kern-.125emX}}

\usepackage[utf8]{inputenc}
\usepackage{amsmath, amsfonts, amsthm, bm, varwidth, xfrac, textgreek, mathtools, tabu, array, float, xcolor, mdframed}
\usepackage{caption}
\usepackage{algorithm}
\usepackage{booktabs}

\usepackage{flushend}
\usepackage{subfigure}
\usepackage{xcolor}
\usepackage{enumitem}
\usepackage{changes}

\definecolor{darkred}{rgb}{0.5, 0, 0}

\usepackage[noend]{algpseudocode}

\usepackage{tabu}
\usepackage[n,advantage,operators,sets,adversary,landau,probability,notions,logic,ff,mm,primitives,events,complexity,asymptotics,keys]{cryptocode}

\usepackage[functional,reproduced]{usenixbadges}

\newenvironment{MyEnumerate}[1]{\begin{enumerate}\setlength{\itemsep}{0.1cm}
\setlength{\parskip}{-0.05cm} \setlength{\leftmargin}{0pt} #1}{\end{enumerate}}
\newenvironment{MyItemize}[1]{\begin{itemize}\setlength{\itemsep}{0.1cm}
\setlength{\parskip}{-0.05cm} \setlength{\leftmargin}{0pt} #1}{\end{itemize}}

\usepackage{tikz}
\usepackage{multirow}
\usetikzlibrary{calc}

\usepackage{mathtools}

\newcommand\DoubleLineTaggedd[7][5pt]{%
  \path(#2)--(#3)coordinate[at start](h1)coordinate[at end](h2);
  \coordinate (h1l) at ($(h1)!#1!10:(h2)$);
  \coordinate (h1r) at ($(h1)!#1!-10:(h2)$);
  \coordinate (h2l) at ($(h2)!#1!-10:(h1)$);
  \coordinate (h2r) at ($(h2)!#1!10:(h1)$);
  \draw[#4] (h1l) -- (h2l);

  \coordinate (h1la) at ($(h1)!15pt!90:(h2)$);
  \coordinate (h1ra) at ($(h1)!15pt!-90:(h2)$);
  \coordinate (h2la) at ($(h2)!15pt!-90:(h1)$);
  \coordinate (h2ra) at ($(h2)!15pt!90:(h1)$);
  \node at ($(h1la)!0.9!(h2la)$) [] {{#6}};
  \node at ($(h1ra)!0.9!(h2ra)$) [] {{#7}};
}

\newcommand{\Adv}{\mathcal{A}}
\newcommand{\Sim}{\mathcal{S}}
\newcommand{\cF}{\mathcal{F}}
\newcommand{\cG}{\mathcal{G}}
\newcommand{\real}{\small \mbox{\sf REAL}}
\newcommand{\ideal}{\small \mbox{\sf IDEAL}}
\newcommand{\bit}{{\{0,1\}}}

\theoremstyle{default}
\newtheorem{theorem}{Theorem}[section]

\newtheorem{corollary}[theorem]{Corollary}

\newtheoremstyle{break}
  {\topsep}{\topsep}%
  {}{}%
  {\bfseries}{.}%
  {\newline}{}%
\theoremstyle{break}
\algrenewcommand\textproc{}
\newtheorem{functionality}[theorem]{FUNCTIONALITY}
\newtheorem{protocol}[theorem]{PROTOCOL}

\renewcommand{\paragraph}[1]{{\medskip \noindent{\bf #1}}}

\newcommand{\picomp}{{\Pi_{\sf comp}}}
\newcommand{\fcomp}{{{\cal F}_{\sf comp}}}

\newcommand{\fbtc}{{\cal F}_{\sf B2C}}
\newcommand{\fctc}{{\cal F}_{\sf C2C}}
\newcommand{\fmult}{{{\cal F}_{\sf MC}}}
\newcommand{\pimulti}{{\Pi_{\sf MC}}}
\newcommand{\fminbtc}{{{\cal F}_{\min}^{\sf B2C}}}
\newcommand{\piminbtc}{{\Pi}_{\min}^{\sf B2C}}
\newcommand{\fmincom}{{\cal F}_{\sf min}^{\sf com}}
\newcommand{\pimincom}{{\Pi}_{\sf min}^{\sf com}}

\begin{document}


\makeatletter
\let\@fnsymbol\@arabic
\makeatother
\title{Prime Match: A Privacy-Preserving Inventory Matching System\footnotemark[1]}


\author{Antigoni Polychroniadou
\and
Gilad Asharov\footnotemark[2]\and
\and
Benjamin Diamond\footnotemark[3]
\and
Tucker Balch
\and
Hans Buehler\footnotemark[3]
\and
Richard Hua
\and
Suwen Gu\\\\ J.P.~Morgan
\and
Greg Gimler\footnotemark[3]
\and
Manuela Veloso
}

\maketitle

\footnotetext[1]{This is the full version of~\cite{PolychroniadouA23}.}

\footnotetext[2]{Department of Computer Science, Bar-Ilan University. Work initiated while at J.P.~Morgan AI Research. Editorial work while at Bar-Ilan University, supported by J.P.~Morgan Faculty Research Award.}

\footnotetext[3]{Authors contributed while  employed at J.P.~Morgan. Buehler moved to XTX Markets, Diamond to Coinbase, and Gimler to Meta. }

\newcommand{\share}[1]{\langle#1\rangle^q}
\begin{abstract}
Inventory matching is a standard mechanism/auction for trading financial stocks by which buyers and sellers can be paired. In the financial world, banks often undertake the task of finding such matches between their clients. The related stocks can be traded without adversely impacting the market price for either client. If matches between clients are found, the bank can offer the trade at advantageous rates. If no match is found, the parties have to buy or sell the stock in the public market, which introduces additional costs. 

A problem with the process as it is presently conducted is that the involved parties must share their order to buy or sell a particular stock, along with the intended quantity (number of shares), to the bank. 
Clients worry that if this information were to “leak” somehow, then other market participants would become aware of their intentions and thus cause the price to move adversely against them before their transaction finalizes.

We provide a solution, Prime Match, that enables clients to match their orders efficiently with reduced market impact while maintaining privacy. In the case where there are no matches, no information is revealed. Our main cryptographic innovation is a two-round secure \emph{linear} comparison protocol for computing the minimum between two quantities without preprocessing and with malicious security, which can be of independent interest. We report benchmarks of our Prime Match system, which runs in production and is adopted by a large bank in the US -- J.P. Morgan. 
The system is designed utilizing a star topology network, which provides clients with a centralized node (the bank) as an alternative to the idealized assumption of point-to-point connections, which would be impractical and undesired for the clients to implement in reality.

Prime Match is the first secure multiparty computation solution running live in the traditional financial world.
\end{abstract}

\section{Introduction}

An axe is an interest in a particular stock that an investment firm wishes to buy or sell. Banks and brokerages provide their clients with a matching service, referred to as ``axe matching". 
%
%
%
%
When a bank finds two clients interested in the same stock but with opposite directions (one is interested in buying and the other is interested in selling), the bank can offer these two clients the opportunity to trade internally without impacting the market price. Both clients, and the bank, benefit from this internalization. On the other hand, if the bank cannot find two matching clients, the bank has to perform the trade in the public market, which introduces some additional costs and might impact the price. Banks, therefore, put efforts into locating internalized matches.  
 


One such effort is the following service. To incentivize clients to trade, 
banks publish a list of stocks that they are interested in trading, known as ``axe list". 
The axe list that the bank publishes contains, among other things, aggregated information on previous transactions that were made by clients and facilitated by the bank. For instance, to facilitate clients' trades, the bank sometimes buys stocks that some clients wish to sell. The bank then looks to sell those stocks to other clients at advantageous rates before selling those stocks in the public market. Those stocks will appear in the bank's axe list. 

The axe list  consists of tuples $({\sf op}, {\sf symb}, {\sf axe})$ where 
${\sf op} \in \{{\sf buy},{\sf sell}\}$,
$\sf symb$ is the symbol of the security to buy or sell,  and $\sf axe$ is the number of the shares (quantity) of the security to buy or sell (we sometimes use the terminology of ``long" for buy and ``short" for sell). This axe list provides clients the ability to locate available (synthetic) trades at reduced financing rates.

Unfortunately, this method is unsatisfactory. Although the information in the axe list of the bank relates to transactions that were already executed, there is a correlation between previous transactions that a client performed and future transactions that it might wish to trade. Thus, clients feel uncomfortable with seeing their recent (potentially large) trade history (although anonymized and aggregated) in the axe list that the bank publishes, and sometimes ask the bank to remove their previous trades from the axe list. Clients, therefore, face the following dilemma: keeping their axes published reveal information about their future potential trades or investment strategy, while continuously asking to remove trades from the axe list limits the banks' ability to internalize trades and offer advantageous rates, to begin with. 

The bank currently uses some ad-hoc methods to mitigate the leakage. For instance, it might aggregate several stocks together into ``buckets" (e.g., reveal only the range of available stocks to trade in some sector), or trim the volumes of other stocks. 
This does not guarantee privacy, and also makes it harder to locate potential matches. 

%
%
%


\subsection{Our Work} 

We provide a novel method for addressing the inventory matching problem (a simple double auction, which is periodic, with a single fixed price per symbol). 
Our main contribution is a suite of cryptographic protocols for several variants of the inventory matching problem. The system we report, called Prime Match, was implemented and runs in production in J.P. Morgan since September 2021. 
Prime Match has the potential to transform common procedures in the financial world. 
We design the following systems:
\begin{itemize}[nosep,leftmargin=10pt]
\item {\it Bank-to-client inventory matching}: Prime Match supports a secure two-party (bank-to-client) inventory matching. The client can privately find stocks to trade in the bank's full axe list without the bank revealing its axe list, and without the client revealing the stocks and quantities it wishes to trade. The protocol is secure against a semi-honest bank and malicious client and is of two rounds of communication (three rounds if both parties learn the output).  

\item {\it Client-to-client inventory matching}: We extend Prime Match to support a secure (client-to-client) inventory matching. This is a three-party protocol where the bank is an intermediate party that mainly delivers the messages between two clients and facilitates the trade if there is a match. This enables two clients to explore whether they can have potential matches against each other and not just against the axe list of the bank. This further increases potential matches. The protocol is secure in the presence of one malicious corruption and is of three rounds of interaction. 

\item \emph{Multiparty inventory matching:} We also extend the client-to-client inventory matching to multiple clients coming at once and looking to be matched.




\end{itemize}


\noindent 
We expand on each one of those scenarios below.

\paragraph{\bf Bank-to-client inventory matching:}
We replace the current procedure in which the bank sends an axe list to a client, and the client replies with which stocks to trade based on the axe list, with a novel bank-to-client inventory matching. Prime Match allows the bank to locate potential matches without revealing its axe list, and without the client revealing its interests. Moreover, as the bank can freely use accurate axe information (as the axe list is hidden), clients have no longer an interest to remove themselves from the axe list. All parties enjoy better internalization and advantageous rates. 


Importantly, the bank does not learn any information about what the client is interested in on any stock that is not matched, and likewise, the client does not learn any information on what is available unless she/he is interested in that as well. Only after matches are found, the bank and the client are notified and the joint interest is revealed. At a high level, for two orders $({\sf buy},{\tt X}, {\sf axe_1})$ and $({\sf sell},{\tt X}, {\sf axe_2})$ on the same symbol ${\tt X}$, we provide a secure two-party protocol that computes as the matching quantity the min quantity between $\sf axe_1$ and $\sf axe_2$.\footnote{In our actual protocol, each party also provides a range of quantities it wishes to trade, i.e., a minimum amount and a maximum amount. If there is no match that satisfies at least its minimum quantity, then there is no trade. To keep the introduction simple, we omit this additional complexity for now.}

\begin{figure}[h]
\begin{center}

\resizebox{\columnwidth}{!}{%
\begin{tikzpicture}

\tikzstyle{bank}=[rectangle, draw, inner sep = 9pt, text badly centered, rounded corners = 3, fill=black!10]
\tikzstyle{agent}=[circle,draw,minimum width=0.9cm]
[auto,
	every node/.style={node distance=4cm},
	.style={rectangle, draw, dotted, inner sep= 2pt, text width=2.2cm, text height=5.2cm},
	cn/.style={rectangle, draw, inner sep = 9pt, text badly centered, rounded corners = 3, fill=black!10},
	as/.style={circle, draw, inner sep = 2pt, text badly centered},
	dcbox/.style={draw, dotted, thick, inner sep = 6pt},
	>=stealth,
	]
\node (O) at (0,0) [bank] {Secure Matching Engine};
\node (k) at (0,0.8) {$\mathsf{order_B}$};
\foreach \x in {1,...,1} {
	\node (C\x) at (180*\x:6cm) [agent] {$C_\x$};
}
\foreach \x in {1,...,1} {
	\DoubleLineTaggedd{O}{C\x}{<-}{->}{$\mathsf{order_C}$}{$ $};
}

\end{tikzpicture}
}
\caption{\small Client-to-bank topology. Client $C$ sends an encrypted order $\mathsf{order_C}=({\sf buy},{\tt X}, {\sf axe_1})$ to the Bank (secure matching engine) which holds $\mathsf{order_B}=({\sf sell},{\tt X}, {\sf axe_2})$. The engine computes the minimum between $\sf axe_1$ and $\sf axe_2$.}
\label{fig:enc1}
\end{center}
\end{figure}
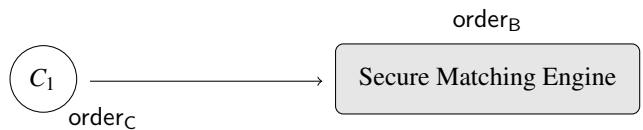

\paragraph{\bf Client-to-client inventory matching:}
The above approach only enables matching between the bank's inventory to each client separately but does not allow a direct matching among different clients. For illustration, consider the following scenario: Client A is interested in buying 100 shares of some security ${\tt X}$, while client B is interested in selling 200 shares of the same security ${\tt X}$. On the other hand, the bank does not provide ${\tt X}$ in its inventory axe list. The bank either distributes in a non-private way its axe list to clients A and B (as it is being conducted prior to our work) or engages twice in a bank-to-client inventory matching described above, the first time against client A and the second time against client B. The two clients do not find ${\tt X}$ in the list, and both clients would have to trade on the public market at higher costs. 

Prime Match allows the clients and the bank not to miss such opportunities. We provide a mechanism that acts as a transparent matching engine. Each client provides as input to the computation his/her encrypted axes, and the clients then interact and learn whether their axes match or not, see Figure \ref{fig:enc2}. 
For this solution, we provide a three-party secure minimum protocol $\Pi_{\mathsf{min}}$ among two clients and the bank as the intermediary party to facilitate and execute the trade if there is a match.

\begin{figure}[h]
\begin{center}
\resizebox{\columnwidth}{!}{%

\begin{tikzpicture}

\tikzstyle{bank}=[rectangle, draw, inner sep = 9pt, text badly centered, rounded corners = 3, fill=black!10]
\tikzstyle{agent}=[circle,draw,minimum width=0.9cm]
[auto,
	every node/.style={node distance=4cm},
	.style={rectangle, draw, dotted, inner sep= 2pt, text width=2.2cm, text height=5.2cm},
	cn/.style={rectangle, draw, inner sep = 9pt, text badly centered, rounded corners = 3, fill=black!10},
	as/.style={circle, draw, inner sep = 2pt, text badly centered},
	dcbox/.style={draw, dotted, thick, inner sep = 6pt},
	>=stealth,
	]
\node (O) at (0,0) [bank] {Secure Matching Engine};
\foreach \x in {1,...,2} {
	\node (C\x) at (180*\x:6cm) [agent] {$C_\x$};
}
\foreach \x in {1,...,2} {
	\DoubleLineTaggedd{O}{C\x}{<-}{->}{$\mathsf{order}_\x$}{$ $};
}

\end{tikzpicture}
}
\caption{\small Client-to-client topology. Client $C_1$ and $C_2$ send encrypted orders $\mathsf{order_1}=({\sf buy},{\tt X}, {\sf axe_1})$ and $\mathsf{order_2}=({\sf sell},{\tt X}, {\sf axe_2})$, respectively, to the Bank which computes the minimum of $\sf axe_1$ and $\sf axe_2$.}
\label{fig:enc2}
\end{center}
\end{figure}
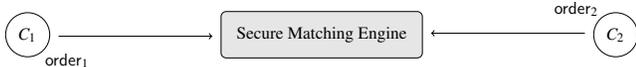

\paragraph{Multi-party protocol.}
A potentially powerful mechanism would be to support multiple clients coming at the same time, where all clients talk to each other through the bank who facilitates the trades when there are matches. This might increase the potential number of matches for each client. We implement such a mechanism based on our client-to-client matching protocol, where we invoke it ${n \choose 2}$ times, for each possible pair of clients, when $n$ is the total number of participating clients. Since the service of axe matching is relatively exclusive, i.e., it is offered only to selected clients, $n$ is relatively small (around $10$ per day), and thus this approach suffices for the current needs.  


At this point, we only implement this relatively degenerated form of multiparty matching. 
We provide security for a semi-honest bank and malicious clients. 
In the multiparty setting, there are further challenges that have to be explored, such as what information is leaked by the functionality due to partial matches (i.e., client A can fulfill its order, say, selling 1000 shares by matching with client B and C, each wishing to buy 500 shares). 
Moreover, to achieve malicious security, the protocol also has to guarantee that the bank does not discriminate against clients, e.g., when two clients are both interested in buying some security, then it treats them fairly and does not prefer to match the ``big client" over the ``small client." In fact, it is impossible to support security against a malicious bank in this case already because of the star network -- the clients communicate through the bank with no authentication (see~\cite{BarakCLPR11}). Therefore, achieving malicious security would require some different setups and further techniques. 
We leave this for future research. 

From a business perspective, the clients generally do trust the bank, and the bank is also highly regulated and will not risk its reputation by attempting to cheat. Therefore, semi-honest bank generally suffices. 

\paragraph{Secure minimum protocol.}
At the heart of our Prime Match engine is a  secure protocol for comparing two input values ${\sf axe_1}$ and ${\sf axe_2}$, each in $\{0, \ldots, 2^n - 1\} \subset \mathbb{F}_q$. The protocol, given the bit-decompositions of ${\sf axe_1}$ and ${\sf axe_2}$, computes the minimum between the two. We have a two-party variant (bank to client) and a three-party variant when only two parties have inputs (client-to-client) and the third party (the server) helps in the computation.
For the latter, an interesting property of our protocol is that the two clients only perform linear operations, and therefore can operate non-interactively on encrypted inputs (or secret-shared, or homomorphic commitments, etc.). The server facilitates the computation. For $\ell$-bit inputs, our protocol runs in three rounds of interaction and with $O(\ell^2)$ communication where in the first round clients provide their input, and in the last (third) round the output is revealed. The protocls also offers malicious security.


%

%

\paragraph{Implementation and evaluation.}
All three scenarios were implemented,  and we report running times in Section~\ref{sec:prime-match-perf}. 
On the bottom level, both bank-to-client and client-to-client protocols can process roughly 10 symbols per second with security against malicious clients under conventional machines with commodity hardware. Our system is running live, in production by J.P. Morgan. To the best of our knowledge, this is the first MPC solution running live in the financial world. Commercially, the main advantage of the system is the increased opportunities for clients to find matches. 

As clients do not wish to spend resources to use such a service (installation of packages, maintenance cost, etc.), and cannot commit to providing tech resources before testing the product, Prime Match is implemented as a browser service. This raises several challenges in the implementation, see Section~\ref{sec:prime-match-perf}. Moreover, in the client-to-client matching a star topology network is required where clients communicate only with the bank. Clients do not wish to establish communication with other clients and reveal their identities to other clients. 

\paragraph{Our contributions.}
To conclude, our contributions are:
\begin{MyItemize}
\item We identify a real-world problem in which cryptography significantly simplifies and improves the current inventory matching procedure. 
\item We provide two new protocols: bank-to-client inventory matching and client-to-client inventory matching. Those completely replace the current method which leaks information and misses potential matches. Our protocols are novel and are specifically tailored to the problem at hand. We do not just use generic, off-the-shelf, MPC protocols (see Section~\ref{sec:into:not-generic-mpc} for a discussion). 
\item At the heart of our matching engine is a novel two-round comparison protocol that minimizes interaction and requires only linear operations. 
\item The protocols are implemented and run live, in production, by a major bank in the US -- J.P. Morgan.
\end{MyItemize}


\subsection{Related Work}
\paragraph{Prior works on volume (quantity) matching.} We now compare the prior privacy-preserving volume matching architectures \cite{Cartlidge:2019aa,FC,Asharov:2020aa} to Prime Match. The MPC-based volume matching constructions of \cite{Cartlidge:2019aa,FC} derive their security by separating the system's \textit{service operator/provider} into several (e.g., 3) distinct servers, whose collusion would void the system's security guarantees. The clients submit their \emph{encrypted} orders to the servers by secret sharing,
such that no single server can recover the encrypted orders. The clients have no control over these servers and no clear way to prevent them from colluding. 

Allowing clients to \textit{themselves} serve as contributing operators of the system would present its own challenges. For instance, it would impose a disproportionate computational burden on those clients who choose to serve as operators. Moreover, it is unclear how to incentivize clients to run heavy computations, and to play the role of the operators.\footnote{Part of the success of the Prime Match system is related to the fact that clients are offered a web service to participate in the system which requires minimal tech support by the clients.} 

The fully homomorphic approach of ~\cite{secretmatch} imposes a computational burden on a single server in a star topology network in which clients communicate with the server. Moreover, the concrete efficiency of the proposed GPU-FHE scheme is slow. Furthermore, the scheme of~\cite{secretmatch} does not offer malicious security. FHE-based solutions for malicious security are much less efficient than the ones based on MPC. 

\paragraph{Prior works on privacy-preserving dark pools.}
A recent line of research has attempted to protect the information contained in dark pools~\cite{Asharov:2020aa,cryptoeprint:2020:662} run by an operator. The systems described in these works allow users to submit orders in an ``encrypted'' form; the markets' operator then compares orders ``through the encryptions'', unveiling them only if matches occur. The functionality of privacy-preserving dark pools is a continuous double auction in which apart from the \textit{direction} (buy or sell) and a desired trading \textit{volume}, a \textit{price} (indicating the ``worst'' price at which the participant would accept an exchange) is submitted. The operator ``matches'' compatible orders, which, by definition, have opposite directions, and for which the price of the buy order (the ``bid'') is at least the price of the sell order (the ``ask''). \cite{Cartlidge:2019aa,cryptoeprint:2020:662,MazloomDPB23} are based on MPC with multiple operators and the work of \cite{Asharov:2020aa} is based on FHE.

Dark pools are different than our setting, as matches are also conditioned on an agreement on a price (requiring many more comparisons) leading  to more complex functionality. In comparison, inventory matching is a simple double auction, which is periodic, with a single fixed price per symbol. 
Moreover, dark pools support high-frequency trading, which means that they have to process orders very fast. All prior works’ performance on dark pools (including multi-server dark pools) does not suffice for high-frequency trading. In comparison, axe-list matching is a much slower process; with the current, insecure procedure of axe-matching, a few minutes might elapse between when the bank sends its axe list, and the time the client submits its orders. Since ensuring privacy introduces some overhead, clients might not necessarily prefer a slower privacy-preserving dark pool over a fast ordinary dark pool. Furthermore, 
secure comparison is a necessary building block for dark pools. Any of the comparison protocols from prior works,
\cite{DBLP:journals/ieicet/NishideO07,Catrina2010ImprovedPF,Wagh:2019aa,DBLP:conf/icalp/LipmaaT13,DBLP:conf/tcc/DamgardFKNT06,DBLP:conf/fc/MakriRVW21}, including ours, can be used for dark pools, but all of them have some overhead. Unfortunately, neither of these works can lead to a fast dark pool (in a star topology network) which is close to the running times of a dark pool operating on plaintexts. Achieving fast enough comparison that is suitable for high-frequency trading is an interesting open problem. 




The work of Massacci et al.~\cite{MassacciNN0W18} considers a distributed Market Exchange for futures assets which has functionality with multiple steps where one of the steps includes the dark pool functionality. Their experiments show that their system can handle up to 10 traders. Moreover, orders are not concealed: in particular, an aggregated list of all waiting buy and sell orders is revealed which is not the case in solution and the dark pool solutions.
%
%
Note that there are works that propose dark pool constructions on the blockchain~\cite{ngo2021practical,galal2021publicly,bag2019seal} which is not the focus of our work. Moreover, these solutions have different guarantees and security goals. None of the above solutions is in production.

\paragraph{Prior works on secure 3-party Less Than comparison.} There are several works in the literature that propose secure comparison protocols of two values in the information-theoretic setting~\cite{DBLP:journals/ieicet/NishideO07,DBLP:conf/tcc/DamgardFKNT06, Catrina2010ImprovedPF,  DBLP:conf/icalp/LipmaaT13,
      Wagh:2019aa,
      DBLP:conf/fc/MakriRVW21,
      cryptoeprint:2022/866}. See Table~\ref{fig:comparison} for a detailed comparison of these works compared to ours. Our protocol does not require preprocessing and runs in 2 rounds of interaction. Our cost incurs an $\ell^2$ overhead since we secret share $\ell$ bit numbers in a field of size $\ell$. Similar overhead also appears in prior works.  The security parameter $\lambda$ overhead is required due to the use of coin flipping and the additional use of commitments in the malicious protocol. The main reason for the higher overhead of prior secret sharing-based protocols in Table~\ref{fig:comparison} is that they require interaction per secure multiplication leading to an increased round complexity ($\approx \log \ell$). Our protocol does not require any secure multiplications, which is a significant benefit in upgrading our passive protocol to one with malicious security.

The works of \cite{{Fischlin01,LinT05,DamgardGK08,DamgardGK09}}, based on multiplicative/additive homomorphic encryption, provide 2 (or constant) round solutions but they only offer passive security.   The computational cost is capped at $O(\lambda\cdot\ell)$ modular multiplications. Moreover, some works require a trusted setup assumption to generate the public parameters. For instance the modulus generation of the homomorphic Paillier encryption-based solutions.

The most recent work of~\cite{cryptoeprint:2022/866}, based on functional secret sharing in the preprocessing model, is a three-round solution offering only passive security with the cost of $O(\ell)$ PRG calls in its online phase.

\begin{table*}[ht]
  \begin{center}
  \begin{footnotesize}
    \begin{tabular}[t]{c|c|c|c|c|c|c|c}
      \toprule
       
     & Offline & Online & Offline & Online  & & \\

   Protocol & Communication  & Communication          & Computation &Computation & Rounds & Security & Corruption\\
      \midrule
 
      \cite{DBLP:journals/ieicet/NishideO07} & - & $O\big((\ell\log \ell)\cdot(\ell+s) \big)$ & - & $O\big((\ell\log \ell)\cdot(\ell+s) \big)$& $31$ & passive &HM\\
      \cite{Catrina2010ImprovedPF} & - & $O\big(\ell\cdot(\ell+s) \big)$ & - & $O\big(\ell\cdot(\ell+s) \big)$ & $O(\log \ell)$& passive & HM\\

\cite{Wagh:2019aa} & - & $O(\ell^2+\log \ell)$ & - & $O(\ell^2+\log \ell)$ & $O(\log \ell)$& passive & DM\\

    This work & - & $O(\ell^2+\log \ell)$ & - & $O(\ell^2+\log \ell)$& 2& passive &DM\\
\hline
      \cite{DBLP:conf/icalp/LipmaaT13} & $O(\ell^2)$ & $O(\log \ell\cdot(\ell+s))$ & $O(\ell^2)$ & $O(\log \ell \cdot(\ell+s))$& $O(\log \ell)$& active  &HM\\
            
            \cite{DBLP:conf/tcc/DamgardFKNT06} & - & $O\big((\ell\log \ell)\cdot(\ell+s) \big)$ & - & $O\big((\ell\log \ell)\cdot(\ell+s) \big)$& $44$ & active &HM\\

      \cite{DBLP:conf/fc/MakriRVW21} & $O(\ell)$ & $O\big((\ell\log \ell)\cdot(\ell+s) \big)$ & $O(\ell)$ & $O(\ell\cdot(\ell+s))$& $O(\log \ell)$& active &DM\\

         This work & - & $O(\ell\cdot(\ell+\lambda))$ & - & $O(\ell\cdot(\ell+\lambda))$&2 & active &HM\\


      \bottomrule
    \end{tabular}
  \end{footnotesize}
  \end{center}
  \caption{
Cost of passive and active comparison protocols in terms of offline, and online communication and computation complexity; in terms of rounds; in terms of security; and in terms of corruptions supported. HM stands for honest majority, while DM stands for dishonest majority.  $\ell$ denotes to the bit length of the input, $s$ is the statistical security parameter and $\lambda$ is the computational security parameter. The work of\cite{DBLP:conf/fc/MakriRVW21} achieves statistical security over arithmetic fields but it achieves perfect security over the arithmetic rings.}
\label{fig:comparison}

\end{table*}%

\subsection{Why Specifically-Tailored Protocols?}
\label{sec:into:not-generic-mpc}
A natural question is why we design a specifically tailored protocol for the system, instead of just using any generic, off-the-shelf secure computation protocols. Those solutions are based on securely emulating arithmetic or Boolean circuits, and require translating the problem at hand to such a circuit. Specifically, for our client-to-client matching algorithm, which is a three-party secure protocol with one corruption, it looks promising to use some generic MPC protocols that are based on replicated secret sharing, such as~\cite{ArakiFLNO16,ChidaGHIKLN18} or garbled circuits~\cite{WangRK17,KatzRR018}. 

There are two main requirements from the system (from a business perspective) that leads us to design a specifically-tailored protocol and not a generic MPC: (1) The need for a constant number of rounds; (2) Working with committed inputs. Furthermore, no offline preprocessing is possible since clients wish to participate only during the live matching phase. We provide a comparison with generic MPC techniques in  Appendix~\ref{sec:generic-mpc}. 

\subsection{Overview of our Techniques}

We focus in this overview on the task of client-to-client matching (see Figure \ref{fig:enc2}): A three-party computation between two clients that communicate through the bank. We present our solution while hiding only the clients' quantities $\sf axe$. However, our detailed protocol additionally hides both the directions and the symbols. We present our protocol in the semi-honest setting and then explain how to achieve malicious corruption.

\medskip
\noindent
{\bf Semi-honest clients and server:} The client provides secret shares (and commitments) for all possible symbols and for the two possible sides. If a client is not interested in buying (resp.~selling) a particular stock, it provides $0$ as its input for that symbol and size. It is assumed that the total number of symbols is around $1000-5000$, and of course, the number of sides is $2$. Thus, each party has to provide roughly $2000-10000$ values. To see if there is a match between clients A and B on a particular stock, we securely compute the minimum between the values the parties provided with opposite sides (i.e., A sells and B buys, or B sells and A buys).

Each one of the clients first secret shares its secret value ${\sf axe}$ using an additive secret sharing scheme. The two clients then exchange shares\footnote{\label{ft:star}The communication model does not allow the two clients to talk directly, and each client talks only to the server. However, using encryption and authentication schemes, the two clients can establish a secure channel while the server just delivers messages for them.}. Then, they decide on the matching quantity by computing two bits indicating whether the two quantities are equal or which one of the two is minimal. 

We design a novel algorithm for computing the minimum. The algorithm consists of two phases. As depicted in Figure~\ref{fig:enc4}, the first phase works on the shares of the two secrets $(\mathsf{a,b})$, exchanged via the matching engine using symmetric key encryption, while performing only linear operations on them ($\mathsf{min}$ protocol). Looking ahead, each one of the two clients would run this phase, without any interaction, on its respective shares. The result would be shares of some secret state ($\mathsf{d_0,d_1}$) in which some additional non-linear processing is needed after reconstruction to obtain the final result. However, the secret state can be simulated with just the result of the computation -- i.e., the two bits indicating whether the two numbers are equal or which one is minimal. Therefore, at the end of the first phase, the two clients can send the shares to the server, who reconstructs the secret state and learns the result, again using just local (this time, non-linear) computations. 

Our minimum protocol $\mathsf{min}$ is described in Section~\ref{sec:comparison}, and we overview our techniques and contributions in the relevant section. Our semi-honest protocol is given in Section~\ref{sec:minimum:semi-honest}.

\medskip
\noindent
{\bf Malicious clients.} We now discuss how to change the protocol to protect against malicious clients. 

Zooming out from computing minimum, the auction works in two phases: a ``registration" stage, where clients submit their orders, and the matching stage, where the clients and the bank run the secure protocols to find matching orders. 
In the malicious case, the parties submit commitments to the quantities of their orders to the server. 
%
%
 The list of participants is not known in advance, only the clients who submitted a commitment can participate in the current matching phase. Moreover, the list of participants (at each run) is not public and is only known to the server. 
 
 Of course, clients have to be consistent, and cannot use different values in the matching phase and in the registration phase. In the matching phase, the clients secret shares their inputs (additive secret sharing), and prove using a Zero-Knowledge (ZK) proof that the shares define the committed value provided in the registration phase. 
 

More specifically, client $C_1$ commits to $\sf a$ during registration, i.e., sends $\sf Com(a)$ to the server and commits to the shares $({\sf a_1,a_2})$ of the minimum by sending $\sf Com(a_1)$ and $\sf Com(a_2)$ to the server. It also proves in ZK the statement that $\sf Com(a)=\sf Com(a_1)+\sf Com(a_2)$ given that the commitment scheme is linearly homomorphic allowing to perform additions on committed values. 


On top of the basic semi-honest protocol (as depicted in Figure~\ref{fig:enc4}) we also exchange the messages shown in Figure~\ref{fig:enc5} where every party forwards a commitment to the other party for the share that it does not hold. Client $C_1$ receives a commitment to $\sf b_2$ and client $C_2$ receives a commitment to $\sf a_1$. 

Next, recall that our minimum protocol requires only linear work from the clients, and thus it allows to work on any linearly-homomorphic cryptosystem, such as linear secret sharing scheme, linearly homomorphic commitments, linearly homomorphic encryption scheme, and so on. In the semi-honest setting, we used this property to work only on the secret shares. We run the linear algorithm three times in parallel, on different inputs:
\begin{MyEnumerate}
\item First, each client simply runs the algorithm on  additive shares, just as in the semi-honest solution. This is depicted in Figure~\ref{fig:enc4}. Running the algorithm on those shares would result in shares of some secret state that will be delivered to the server. The server reconstructs the state and computes the result from this state. 
\item Second, the parties run the algorithm on the commitments of the other party's share. This is depicted in Figure~\ref{fig:enc5}. Since the commitment scheme is also linearly-homomorphic, it enables Alice to compute a commitment of what Bob is supposed to send to the server in the first invocation, and vice versa. 
\item Third, the parties also compute (again, using only linear operations!) information that allows the server to learn the openings of the other party's commitment. This enables the server to check that all values it received in the first invocations are correct.  
\end{MyEnumerate}
%

\begin{figure}[t]
\begin{center}
\resizebox{\columnwidth}{!}{%

\begin{tikzpicture}

\tikzstyle{bank}=[rectangle, draw, inner sep = 9pt, text badly centered, rounded corners = 3, fill=black!10]
\tikzstyle{agent}=[circle,draw,minimum width=0.9cm]
[auto,
	every node/.style={node distance=4cm},
	.style={rectangle, draw, dotted, inner sep= 1pt, text width=2.2cm, text height=5.2cm},
	cn/.style={rectangle, draw, inner sep = 9pt, text badly centered, rounded corners = 3, fill=black!10},
	as/.style={circle, draw, inner sep = 2pt, text badly centered},
	dcbox/.style={draw, dotted, thick, inner sep = 6pt},
	>=stealth,
	]
\node (O) at (0,0) [bank] {Secure Matching Engine};
\foreach \x in {1,...,2} {
	\node (C\x) at (180*\x:6cm) [agent] {$C_\x$};
}
\foreach \x in {1,...,2} {
	\DoubleLineTaggedd{O}{C\x}{<-}{->}{$ $}{$ $};
}
\node (inp1) at (-6,1) {$\sf a={\color{blue}a_1}+ {\color{cyan}a_2}$};
\node (inp2) at (6,1) {$\sf b={\color{blue}b_1}+ {\color{cyan}b_2}$};

\node (msg1) at (-4,0.3) {$\sf {\color{cyan}a_2}$};
\node (msg2) at (4,0.3) {$\sf {\color{blue}b_1}$};

\draw [to-](-5.5,-1) -- (-2.5,-1);
\draw [to-](5.5,-1) -- (2.5,-1);

\node (msg3) at (-4,-0.7) {$\sf {\color{blue}b_1}$};
\node (msg4) at (4,-0.7) {$\sf {\color{cyan}a_2}$};

\node (comp1) at (-6.5,-1) {$\sf {\color{blue}b_1}$};
\node (comp2) at (6.5,-1) {$\sf {\color{cyan}a_2}$};

\draw [-to](-5.5,-2) -- (-2.5,-2);
\draw [-to](5.5,-2) -- (2.5,-2);

\node (msg5) at (-4,-1.7) {${\color{blue}d_1}=\sf min({\color{blue}a_1,b_1)}$};
\node (msg6) at (4,-1.7) {${\color{cyan}d_2}=\sf min({\color{cyan}a_2,b_2})$};

\node (output) at (0,-2) {$d={\color{blue}d_1}+{\color{cyan}d_2}$};

\end{tikzpicture}
}
\caption{\small Client-to-client matching protocol for computing the minimum between the quantities $\sf a$ from client $C_1$ and $\sf b$ from client $C_2$ in the semi-honest setting. As described in Footnote~\ref{ft:star}, the communication between the two clients through the server is encrypted, and so the view of the server in this communication is just ${\color{blue}d_1}$, ${\color{cyan}d_2}$. }
\label{fig:enc4}
\end{center}
\end{figure}
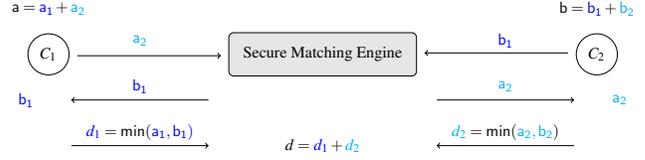
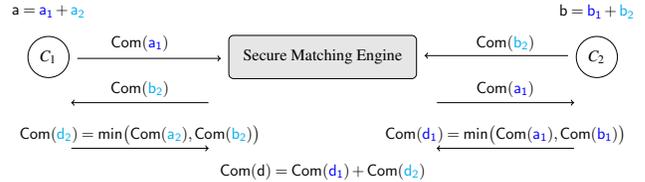
\begin{figure}[t]
\begin{center}
\resizebox{\columnwidth}{!}{%

\begin{tikzpicture}

\tikzstyle{bank}=[rectangle, draw, inner sep = 9pt, text badly centered, rounded corners = 3, fill=black!10]
\tikzstyle{agent}=[circle,draw,minimum width=0.9cm]
[auto,
	every node/.style={node distance=4cm},
	.style={rectangle, draw, dotted, inner sep= 1pt, text width=2.2cm, text height=5.2cm},
	cn/.style={rectangle, draw, inner sep = 9pt, text badly centered, rounded corners = 3, fill=black!10},
	as/.style={circle, draw, inner sep = 2pt, text badly centered},
	dcbox/.style={draw, dotted, thick, inner sep = 6pt},
	>=stealth,
	]
\node (O) at (0,0) [bank] {Secure Matching Engine};
\foreach \x in {1,...,2} {
	\node (C\x) at (180*\x:6cm) [agent] {$C_\x$};
}
\foreach \x in {1,...,2} {
	\DoubleLineTaggedd{O}{C\x}{<-}{->}{$ $}{$ $};
}
\node (inp1) at (-6,1) {$\sf a={\color{blue}a_1}+ {\color{cyan}a_2}$};
\node (inp2) at (6,1) {$\sf b={\color{blue}b_1}+ {\color{cyan}b_2}$};

\node (msg1) at (-4,0.3) {$\sf Com({\color{blue}a_1})$};
\node (msg2) at (4,0.3) {$\sf Com({\color{cyan}b_2})$};

\draw [to-](-5.5,-1) -- (-2.5,-1);
\draw [to-](5.5,-1) -- (2.5,-1);

\node (msg3) at (-4,-0.7) {$\sf Com({\color{cyan}b_2})$};
\node (msg4) at (4,-0.7) {$\sf Com({\color{blue}a_1})$};


\draw [-to](-5.5,-2) -- (-2.5,-2);
\draw [-to](5.5,-2) -- (2.5,-2);

\node (msg5) at (-4,-1.7) {$\sf Com({\color{cyan}d_2})=\sf min\big (Com({\color{cyan}a_2}),\sf Com({\color{cyan}b_2})\big )$};
\node (msg6) at (4,-1.7) {$\sf Com({\color{blue}d_1})=\sf min\big (\sf Com({\color{blue}a_1}),Com({\color{blue}b_1})\big )$};

\node (output) at (0,-2.5) {$\sf Com(d)={\sf Com(\color{blue}d_1})+ \sf Com({\color{cyan}d_2})$};

\end{tikzpicture}
}
\caption{\small Client-to-client matching protocol for computing the minimum in the presence of a malicious adversary. In addition to values computed in Figure~\ref{fig:enc4}, the parties compute commitments of the value that the other participant is supposed to send to the server. }
\label{fig:enc5}
\end{center}
\end{figure}

\medskip
\noindent
{\bf Malicious server.} 
Our final system does not provide security against a malicious server, unless the two clients can authenticate themselves to each other, or can talk directly. We show that if the client can communicate to each other, then we can also support malicious server for the comparison protocol. 


The server receives shares of some secret states, together with commitments of the secret states. It then reconstructs the secret state and checks for consistency. It then has to perform some non-linear operations on the secret state to learn the result. Applying generic ZK proofs for proving that the non-linear operation was done correctly would increase the overhead of our solution. Luckily, the non-linear operation that the server performs is ZK-friendly, specifically, it is enough to prove in ZK a one-out-of-many proof (i.e., given a vector, proving that one of the elements in the vector is zero; see~Theorem~\ref{thm:correctness}). For this particular language, there exists fast ZK solutions~\cite{Groth:2015aa}. See Section \ref{sec:minimum:malicious} for a description and details.

\paragraph{Organization.}
The paper is organized as follows. In Section~\ref{sec:preliminaries} we provide the preliminaries, while some are deferred to the appendices. In Section~\ref{sec:prime-match} we provide the main matching engine functionality. In Section~\ref{sec:minimum} we provide our protocol for computing the minimum, including the semi-honest and the malicious versions. In Section~\ref{sec:prime-match-perf} we report the system performance and in Appendix~\ref{app:imp} we mention challenges pertaining to the deployment of our system.

\section{Preliminaries}
\label{sec:preliminaries}
Some preliminaries are deferred to Appendix~\ref{appx:preliminaries}. 

\paragraph{Notations.}
We use PPT as an acronym for probabilistic polynomial time. We use $\lambda$ to denote the security parameter, and ${\sf negl}(\lambda)$ to denote a negligible function (a function that is smaller than any polynomial for sufficiently large $\lambda$). 



\paragraph{Commitment schemes.}
A commitment scheme is a pair of probabilistic algorithms $(\mathsf{Gen}, \mathsf{Com})$; given public parameters $\mathsf{params} \gets \text{Gen}(1^\lambda)$ and a message $m$, $\mathsf{com} := \mathsf{Com}(\mathsf{params}, m; r)$ returns a ``commitment'' to the message $m$. To reveal $m$ as an opening of $\mathsf{com}$, the committer simply sends $m$ and $r$ (this is sometimes called ``decommitment''). For notational convenience, we often omit $\mathsf{params}$. A commitment scheme is \textit{homomorphic} if, for each $\mathsf{params}$, its message, randomness, and commitment spaces are abelian groups, and the corresponding commitment function is a group homomorphism. We always write message and randomness spaces additively and write commitment spaces multiplicatively. See Section \ref{appx:preliminaries} for more details.

\paragraph{Zero-knowledge proofs.}
We use non-interactive zero-knowledge for three languages. See formal treatment in Appendix~\ref{appendix:sec:zk}:
\begin{MyItemize}
\item {\bf Commitment equality proof:} Denoted as the relation ${\cal R}_{\sf ComEq}$, the prover convinces the verifier that the two given commitments $c_0,c_1$ hide the same value. 
\item {\bf Bit proof:} Denoted as the relation ${\cal R}_{\sf BitProof}$, allows the prover to prove that a commitment $c$ hides a bit, i.e., a value in $\{0,1\}$. 
\item {\bf Out-of-many proofs.}
Denoted as the relation ${\cal R}_{\sf OneMany}$, allows a prover to prove that one of the commitments $V_0,\ldots,V_n$ in the statement is a commitment of $0$. 
\end{MyItemize}

\section{The Prime Match Main Functionalities}
\label{sec:prime-match}

We now describe our Prime Match inventory matching functionalities. We describe the bank-to-client functionality (Section~\ref{sec:func:bank-to-client}), the client-to-client functionality (Section~\ref{sec:func:client-to-client}), and the multi-client system (Section~\ref{sec:func:multiparty}). 

\subsection{Bank to Client Matching}
\label{sec:func:bank-to-client}
This variant is a two-party computation between a bank and a client. The bank tries to find matching orders between its own inventory and each client separately. As mentioned in the introduction, this essentially comes to replace the current procedure of axe-matching as being conducted today, with a privacy-preserving mechanism. Today the bank sends its inventory list to the client who then submits orders to the bank. 
Note, however, that if the bank runs twice with two different clients, and the bank does not hold some security {\tt X}, and two clients are interested in {\tt X} with opposite directions, then such a potential match will not be found.  

\noindent
The functionality proceeds as follows. 
The client sends to the bank its axe list. This includes the list of securities it is interested in, and for each security whether it is interested in long (buy) or short (sell) exposures, and the quantity. The client sends its own list. The functionality finds whether the bank and the client are interested in the same securities with opposite sides, and in that case, it provides as output the 
matching orders and the quantity is the minimum between the two amounts.

\begin{mdframed}

\begin{functionality}[$\mathcal{F}_{\text{B2C}}$--Bank-to-client functionality]
\label{func:b2c}
The functionality is parameterized by the set of all possible securities to trade, a set $U$. 

\noindent
{\bf Input:} The bank $P^*$ inputs lists of orders $({\sf symb}_i^*, {\sf side}_i^*, {\sf amount}_i^*)$ where ${\sf symb}_i \subseteq U$ is the security, ${\sf side}_i^* \in \{{\sf buy}, {\tt sell}\}$ and ${\sf amount}_i^*$ is an integer. The client sends its list of the same format, $({\sf symb}_i^C, {\sf side}_i^C, {\sf amount}_i^C)$.

\noindent
{\bf Output:} Initialize a list of ${\sf Matches}$.
For each $i,j$ such that ${\sf symb}_i^* = {\sf symb}_j^C$ and ${\sf side}_i^* \neq {\sf side}_j^C$, add $({\sf symb}_i^*, {\sf side}_i^*, {\sf side}_j^C, \min\{{\sf amount}_i^*, {\sf amount}_j^C\})$ to $M$. Output $M$ to both parties. 







\end{functionality}
\end{mdframed}

From a business perspective, it is important to note that the input of the client (and the bank) serves as a ``commitment" - if a match is found then it is executed right away.

The functionality resembles a set intersection. In set intersection, if some element is in the input set of some party but not in its output, it can conclude that it does not contain in the other set. Here, if a party does not find a particular symbol in its output although it did provide it as an input, then it is still uncertain whether the other party is not interested in that security, or whether it is interested but with the same side. 
We show how to implement the functionality in the presence of a malicious client or a semi-honest server in Appendix~\ref{sec:twoparty}.

\paragraph{Bank to multiple clients.}
In the actual system, the bank has to serve multiple clients. This is implemented by a simple (sequential) composition of the functionality. Specifically, the functionality is now reactive where clients first register that they are interested to participate. The bank then runs Functionality~\ref{func:b2c} with the clients -- either according to the basis of first-come-first-served, or some random ordering. We omit the details and exact formalism as they are quite natural given a semi-honest bank.

\paragraph{Range functionalities.}
In Appendix~\ref{sec:range}, we show a variant of the protocol where each party inputs a \emph{range} in which it is interested and not just one single value. I.e., if a matched order does not satisfy some minimal value, it will not be executed. Since the minimum value does not change throughout the execution, whenever the bank receives $0$ as a result of the execution it cannot decide whether the client is not interested in that particular symbol, or whether it is interested -- but the matched amount does not satisfy the minimum threshold.

\subsection{Client to Client Matching}
\label{sec:func:client-to-client}
In this variant, the bank has no input and it tries to find potential matches between clients facilitating two clients that wish to compare their inventories. This is a three-party computation where the bank just facilitates the interaction. It is important to notice that the clients do not know each other, and do not know who they are being paired with. The bank selects the two clients and offers them to be paired. 

\begin{mdframed}

\begin{functionality}[$\mathcal{F}_{\text{C2C}}$--Client-to-client functionality]
\label{func:c2c}
The functionality is parameterized by the set of all possible securities to trade, a set $U$. This is a three-party functionality between two clients, $P_1$ and $P_2$, and the bank $P^*$. 

\noindent
{\bf Input:} The client $P_1$ inputs a list of orders $({\sf symb}_i^{1}, {\sf side}_i^{1}, {\sf amount}_i^{1})$. The client $P_2$ inputs a list of orders $({\sf symb}_i^{2}, {\sf side}_i^{2}, {\sf amount}_i^{2})$, and the bank has no input. 

\noindent
{\bf Output:} Initialize a list of ${\sf Matches}$.
For each $i,j$ such that ${\sf symb}_i^{1} = {\sf symb}_j^{2}$ and ${\sf side}_i^{1} \neq {\sf side}_j^{2}$, add $({\sf symb}_i^{1}, {\sf side}_i^{1}, {\sf side}_j^{2}, \min\{{\sf amount}_i^{1}, {\sf amount}_j^{2}\})$ to $M$. Output $M$ to all three parties. 

\end{functionality}
\end{mdframed}
In the next section, we show how to implement this functionality in the presence of a malicious client or a malicious server, assuming that the two clients can communicate directly, or have a public-key infrastructure. When the two clients can communicate only through the server and there is no public-key infrastructure (PKI) or any other setup, there is no authentication and the server can impersonate the other client. We therefore cannot hope to achieve malicious security. We achieve security in the presence of a semi-honest server. See a discussion in the next subsection. 

\subsection{The Multi-Client System}
\label{sec:func:multiparty}
We now proceed to the multiparty auction. Here we have parties that register with their intended lists, and the bank facilitates the orders by pairing the clients according to some random order. 
The functionality is now reactive; The parties first register, in which they announce that they are willing to participate in the next auction, and they also commit to their orders. In the second phase, the bank selects pairs of clients in a random order to perform client-to-client matching. Looking ahead, typically there are around 10 clients that participate in a given auction. 

For simplicity of exposition and to ease readability, we write the functionality as the universe is just a single symbol. Moreover, instead of sending the side explicitly, the client sends two integers $L$ and $S$, representing its interest in long (buy) or short (sell) exposure, respectively. Rationally, each party would put one of the integers as $0$ (as otherwise, it would just pay extra fees to the bank). Generalizing the functionality to deal with many symbols is performed in a natural manner, where the number of total symbols is 1000-5000 in practice. The main functionality can process all the different symbols in parallel. 

\begin{mdframed}
\begin{functionality}[$\fmult$ -- Multi-client matching]
\label{func:multi}
This is an $n+1$ party functionality between $n$ clients $P_1,\ldots,P_n$ and a bank $P^*$. 

Upon initialization, $\fmult$ initializes a list $\cal P=\emptyset$ and two vectors $\mathcal{L}$ and $\mathcal{S}$ of size $n$, where $n$ bounds the total number of possible clients. 

\medskip
\noindent
{\bf $\fmult.\textsf{Register}(P_i, L_i, S_i)$.} 
Store $\mathcal{L}[i] = L_i$ and $\mathcal{S}[i] = S_i$ and add $i$ to $\cal P$. 
Send to the bank $P^*$ the message ${\sf registered}(P_i)$. 

\noindent
{\bf $\fmult.\textsf{Process}()$}.
\begin{enumerate}[nosep]
\item Choose a random ordering $O$ over all pairs of $\cal P$. 
\item For the next pair $(i,j) \in O$ try to match between $P_i$ and $P_j$ (we can assume wlog that always $i \leq j$): 
\begin{enumerate}[leftmargin=*]
\item Compute $M_0 = {\sf min}({\cal L}[i], {\cal S}[j])$, $b_0^0 = ({\cal L}[i] \leq {\cal S}[j])$, $b_1^0 = ({\cal S}[i] \leq {\cal L}[j])$. 
\item Compute $M_1 = {\sf min}({\cal S}[i], {\cal L}[j])$, $b_0^1 = ({\cal S}[i] \leq {\cal L}[j])$ and $b_1^1 = ({\cal S}[i] \leq {\cal L}[j])$. 
\item Send $(i,j,M_0,M_1, b_i^0,b_i^1)$ to $P_i$, and $(i,j,M_0,M_1)$ with $(b_0^0,b_1^0,b_0^1,b_1^1)$ to $P^*$.
\item Update ${\cal L}[i] = {\cal L}[i]-M_0$ and ${\cal S}[j] = {\cal S}[j]-M_0$. 
\item Update ${\cal S}[i] = {\cal S}[i]-M_1$ and ${\cal L}[j] = {\cal L}[j]-M_1$.
\end{enumerate}
\end{enumerate}

%
%
%

%
%
%
%
%
\end{functionality}
\end{mdframed}

\paragraph{On malicious server.}
Our final protocol (see Appendix~\ref{sec:mainprotocol}) for $\fmult$ is secure in the presence of a semi-honest $P^*$ (and a malicious client). Inherently, clients communicate through a star network where the bank facilitates the communication. Moreover, we assume no PKI, clients do not know how many clients are registered in the system, and how many clients are participating in the current auction. 
This can be viewed as ``secure computation without authentication", in which case the server can always ``split" the communication and disconnect several parties from others (see~\cite{BarakCLPR11} for a formal treatment).

We prove security in the presence of a semi-honest server. In fact, our protocol achieves a stronger notion of a guarantee than just semi-honest, as in particular, it runs the underlying comparison protocol (a single invocation of a client-to-client matching) which is secure against a malicious server.  

Another relaxation that we make is that the ordering of pairs is random, and we do not have a mechanism to enforce it. 
 Note also that the functionality leaks some information to the server; in particular, after finding a match, the bank executes it immediately. The bank can infer information about whether two values equal to $0$, and therefore whether a client is not interested in a particular symbol. In contrast, each client just learns whether its value is smaller or equal to the value of the other party, and therefore when it inputs $0$ it can never infer whether the other party is interested in that symbol or not.


%
%

\section{Securely Computing Minimum}
\label{sec:minimum}
A pivotal building block in Prime Match is a secure minimum protocol. In Section~\ref{sec:minimum:func}, we review our functionality for computing the minimum. We focus on the case of client-to-client matching with an aiding server. We show how to convert the protocol for two parties in Appendix~\ref{sec:twoparty}. 

In Section~\ref{sec:comparison} we present the underlying idea for computing the minimum. The algorithm computes the minimum while using only linear operations (looking ahead, those would be computed on shared values) while pushing the non-linear operations on reconstructed data. In Sections~\ref{sec:minimum:semi-honest} and ~\ref{sec:minimum:malicious} we show a semi-honest and a malicious protocol for computing the minimum, respectively. 

\subsection{The Minimum Functionality}
\label{sec:minimum:func}

After receiving a secret integer from each one of the two parties, the functionality compares them and gives as a result two bits -- which indicate which one of the two inputs is smaller than the other, or whether they are equal. It also gives the result to the server.  

\begin{mdframed}
\begin{functionality}[$\fcomp$: Server-aided secure minimum functionality] \label{func:compare}
Consider two players, $P_0$ and $P_1$, and a server $P^*$. 
\begin{itemize}[nosep,leftmargin=10pt]
\item {\bf Input:} $P_0$ and $P_1$ respectively send integers $v_0$ and $v_1$ in $\{0, \ldots , 2^n - 1\}$ to $\fcomp$.
\item {\bf Output:} $\fcomp$ sends $b_0:=(v_0\leq v_1)$ to $P_0$,  $b_1:=(v_1\leq v_0)$ to $P_1$, $(b_0,b_1)$ to $P^*$. 
\end{itemize}
\end{functionality}
\end{mdframed}


In the rest of this section, we will show how to implement this functionality in the presence of a semi-honest (Section~\ref{sec:minimum:semi-honest}) and malicious adversary (Section~\ref{sec:minimum:malicious}).\footnote{For ease of presentation, Functionality~\ref{func:compare} is for the semi-honest version of the protocol; For the malicious case, we will use a slightly different functionality on committed inputs; See Appendix~\ref{sec:malicious-proof}.} 

\subsection{Affine-Linear Comparison Function}\label{sec:comparison}

We first describe an abstract algorithm which compares two elements $v_0$ and $v_1$ of $\{0, \ldots , 2^n - 1\} \subset \mathbb{F}_q$, given their bit-decompositions. 
We separate the algorithm into two parts: {\sf ComparisonInitial} (Algorithm~\ref{alg:initial}) and ${\sf ComparisonFinal}$ (Algorithm~\ref{alg:final}). Both parts do not use any underlying cryptographic primitives. 

In the first algorithm ({\sf ComparisonInitial}), all operations on the bit-decompositions of the two inputs $v_0$ and $v_1$ are \emph{linear}. Looking ahead, this will be extremely useful when converting the algorithm into a secure two-party protocol, where $v_0$ and $v_1$ are additively shared between the two parties (or also just committed, encrypted under additively homomorphic encryption scheme, etc.). In particular, this part of the protocol can be executed without any interaction, just as the algorithm itself when $v_0$ and $v_1$ are given in the clear. The second algorithm ({\sf ComparisonFinal})  can be computed by a \emph{different} party, given all information in the clear. Looking ahead, this will be executed by the server $P^*$ on the outputs of the first part. This part contains some non-linear operations, however, this part of the algorithm does not have to be translated into a secure protocol. 

\paragraph{Overview Algorithm~\ref{alg:initial} ({\sf ComparisonInitial}).}
Our approach is inspired by the algorithm of Wagh, Gupta, and Chandran~\cite[Alg.~3]{Wagh:2019aa}, which compares a \textit{secret-shared} integer with a \textit{public} integer. (Specifically, its inputs consist of an array of public bits and an array of secret-shared bits.) We extend the algorithm and allow the comparison of two private integers using only linear operations. 

We achieve this $\mathbb{F}_q$-linearity in the following way. We fix integers $v_0$ and $v_1$ in $\{0, \ldots , 2^n - 1\}$, with big-endian bit-decompositions given respectively by \vspace{-1ex}
$$
v_0 = \sum_{j = 0}^{n - 1} 2^{n - 1 - j} \cdot v_{0, j}, {\rm~~~~and~~~~} v_1 = \sum_{j = 0}^{n - 1} 2^{n - 1 - j} \cdot v_{1, j} \ .\vspace{-1ex}
$$ We follow the paradigm of \cite{Wagh:2019aa}, whereby, for each $j \in \{0, \ldots , n - 1\}$, a quantity $w_j$ is computed which equals 0 if and only if $v_{0, j} = v_{1, j}$. Meanwhile, for each $j \in \{0, \ldots , n - 1\}$, we set $c_j := 1 + v_{0, j} - v_{1, j} + \sum_{k < j} w_k$ (we also set $c_n := \sum_{k = 0}^{n - 1} w_k$, as we discuss below). The crucial observation of \cite{Wagh:2019aa} is that, for each $j \in \{0, \ldots , n - 1\}$, $c_j = 0$ so long as $v_{0, j} < v_{1, j}$ as bits (that is, if $1 + v_{0, j} - v_{1, j} = 0$) \textit{and} the higher bits of $v_0$ and $v_1$ agree (inducing the equality $\sum_{k < j} w_k = 0$). By consequence, \textit{some} $c_j$, for $j \in \{0, \ldots , n - 1\}$, must equal 0 whenever $v_0 < v_1$. Similarly, $c_n = 0$ whenever $v_0 = v_1$. In summary, $v_0 \leq v_1$ implies that some $c_j = 0$, for $j \in \{0, \ldots , n \}$.

The main challenge presented by this technique is to ensure that the opposite implication holds; that is, we must prevent the sum $c_j := 1 + v_{0, j} - v_{1, j} + \sum_{k < j} w_k$ from equalling 0 (possibly by overflowing) modulo $q$---that is, even when $w_k \neq 0$ for some $k < j$---and hence yielding a ``false positive'' $c_j = 0$, which would falsely assert the inequality $v_0 \leq v_1$. \cite{Wagh:2019aa} prevents this phenomenon by ensuring that each $w_j \in \{0, 1\}$, and choosing $2 + n < q$ (they set $n = 64$ and $q = 67$). In fact, \cite{Wagh:2019aa} defines $w_j := (v_{0, j} - v_{1, j})^2$. Under this paradigm, $c_j := 1 + v_{0, j} - v_{1, j} + \sum_{k < j} w_k$ is necessarily \textit{non-zero} so long as either $v_{0, j} \geq v_{i, j}$ as bits (so that $1 + v_{0, j} - v_{1, j} > 0$) or \textit{any} $w_k \neq 0$, for $k < j$.

This squaring operation is nonlinear in the bits $v_{0, j}$ and $v_{1, j}$, and so it is unsuitable for our setting. We adopt the following recourse instead, which yields $\mathbb{F}_q$-linearity at the cost of requiring that the number of bits $n \in O(\log q)$ (a mild restriction in practice). The key technique is that we may eliminate the squaring---thereby allowing each $w_j$ to remain in $\{-1, 0, 1\}$---provided that we multiply each $w_j$ by a suitable public scalar. In fact, it suffices to multiply each (unsquared) difference $w_j$ by $2^{2 + j}$. In Theorem~\ref{thm:correctness} below, we argue that this approach is correct.

Our modifications to \cite{Wagh:2019aa} also include our computation of the \textit{non-strict} inequality $v_0 \leq v_1$---effected by the extra value $c_n$---as well as our computation of the opposite non-strict inequality, $v_1 \leq v_0$, in parallel. The latter computation proceeds identically, except uses $-1 + v_{0, j} - v_{1, j} \in \{-2, -1, 0\}$ at each bit.
\begin{algorithm}
\caption{$\textsf{ComparisonInitial}\left( \left(v_{0, 0}, \ldots , v_{0, n - 1}\right),\right.\\~ \left.\hspace{+0.6\columnwidth}\left(v_{1, 0}, \ldots , v_{1, n - 1}\right)\right)$}
\label{alg:initial}
\begin{algorithmic}[1]
\State Assign $w_{\mathsf{accum}} := 0$
\For{$j \in \{0, \ldots , n - 1\}$}
\State Set $c_{0, j} := 1 +  v_{0, j} - v_{1, j} + w_{\mathsf{accum}}$.
\State Set $c_{1, j} := -1 + v_{0, j} - v_{1, j} + w_{\mathsf{accum}}$
\State Set $w_j := \left(v_{0, j} - v_{1, j}\right)$
and  $w_{\mathsf{acum}} \mathrel{+}= 2^{2 + j} \cdot w_j$
\EndFor
\State Set $c_{0, n}$ and $c_{1, n}$ equal to $w_{\mathsf{accum}}$
\State Sample a random permutation $\pi \gets \mathbf{S}_{n + 1}$ 
\For{$j \in \{0, \ldots , n\}$}
\State Sample random scalars $s_{0, j}$, $s_{1, j}$ from $\mathbb{F}_q\setminus \{0\}$.
\State Assign $d_{0, j} := s_{0, j} \cdot c_{0, \pi(j)}$,
\State Assign $d_{1, j} := s_{1, j} \cdot c_{1, \pi(j)}$
\EndFor
\State \Return $(d_{0, 0}, \ldots , d_{0, n})$ and $(d_{1, 0}, \ldots , d_{1, n})$
\end{algorithmic}
\end{algorithm}

Of course, the intermediate value $v_{0, j} - v_{1, j}$ need only be computed once per iteration of the first loop. %

\paragraph{Overview of Algorithm~\ref{alg:final} ({\sf ComparisonFinal}).}
Note that in Algorithm~\ref{alg:initial}, $w_{\mathsf{acum}}=0$ as long as $v_{0,j}=v_{1,j}$, and it attains a non-zero value at the first $j$ for which $v_{0, j} \neq v_{1, j}$. Up to that point, $(c_{0,j},c_{1,j})=(1,-1)$. At the first $j$ for which $v_{0,j}\neq v_{1,j}$:
\begin{MyItemize} 
\item If $v_0 > v_1$ (i.e., $(v_{0,j},v_{1,j})=(1,0)$), then we get that $(c_{0,j},c_{1,j})=(2,0)$. 
\item If $v_{0} < v_1$ (i.e., $(v_{0,j},v_{1,j})=(0,1)$), then we get that $(c_{0,j},c_{1,j})=(0,-2)$. 
\end{MyItemize}
The algorithm then makes sure that no other value of $c_{b,j'}=0$, essentially by assigning $w_{\mathsf{acum}}$ to be non-zero. If $v_0=v_1$ then $c_{0,n}=c_{1,n}=0$. Finally, all the bits $(c_{0,0},\ldots,c_{0,n}), (c_{1,0},\ldots,c_{1,n})$ are permuted and re-randomized with some random scalars. Observe that if $v_0 > v_1$ then all the values $d_{0,0},\ldots,d_{0,n}$ are all non-zero, and one of $d_{1,0},\ldots,d_{1,n}$ is zero. If $v_1 \geq v_0$ then exactly one of the values $d_{1,0},\ldots,d_{1,n}$ is $0$ and all $d_{0,0},\ldots,d_{0,n}$ are non-zero. 

It is crucial that the vectors $(d_{0,0},\ldots,d_{0,n})$ and $(d_{1,0},\ldots,d_{1,n})$ do not contain any information on $v_0, v_1$ rather then whether $v_0 \leq v_1$ or $v_1 \leq v_0$. Specifically, these values can easily be simulated given just the two bits $v_0 \leq v_1$ and $v_1 \leq v_0$. Therefore, it is safe to give both vectors to a third party, which will perform the non-linear part of the algorithm. For a vector of bits $(x_0,\ldots,x_n)\in\bit^{n+1}$, the operation $\textbf{any}_{j = 0}^n x_j$ returns $1$ iff there exists $j \in [0,\ldots,n]$ such that $x_j=1$. Algorithm~\ref{alg:final} simply looks for the $0$ coordinate in the two vectors. We have:

\begin{figure*}
\begin{mdframed}
\begin{protocol}[Semi-honest secure comparison protocol] \label{prot:semi-honest}
\leavevmode

\begin{itemize}[nosep,leftmargin=15pt]
\item {\bf Input:} $P_0$ and $P_1$ hold integers $v_0$ and $v_1$, respectively, in $\{0, \ldots , 2^n - 1\}$. $P^*$ has no input.

\item \textbf{The protocol:}
\begin{enumerate}[nosep,leftmargin=15pt]
\item $P_0$ and $P_1$ engage in the coin-tossing procedure $\Pi_{\sf CT}$ (see Sect.~\ref{appx:preliminaries}) to obtain a $\lambda$-bit shared secret~$s$. 

\item Each party $P_i$ (for $i \in \bit$) computes the bit decomposition $v_i = \sum_{j = 0}^{n - 1} 2^{n - 1 - j} \cdot v_{i, j}$, for bits $v_{i, j} \in \{0, 1\}$.
\item  For each $j \in \{0, \ldots , n - 1\}$, $P_i$ computes a random additive secret-sharing $v_{i, j} = \share{v_{i, j}}_0 + \share{v_{i, j}}_1$ in $\mathbb{F}_q$.  $P_i$ sends the shares $\left( \share{v_{i, j}}_{1 - i} \right)_{j = 0}^{n - 1}$ to $P_{1 - i}$.
\item After receiving the shares $\left( \share{v_{1 - i, j}}_i \right)_{j = 0}^{n - 1}$, from $P_{i - 1}$, $P_i$ executes Algorithm~\ref{alg:initial} on the appropriate shares; that is, it evaluates
\begin{equation*}
\left( \left( \share{d_{0, j}}_i \right)_{j = 0}^n, \left( \share{d_{1, j}}_i \right)_{j = 0}^n \right) \gets \mathsf{ComparisonInitial}\left( \left( \share{v_{0, j}}_i \right)_{j = 0}^{n - 1}, \left( \share{v_{1, j}}_i \right)_{j = 0}^{n - 1} \right) \ ,
\end{equation*}
where all internal random coins are obtained from $G(s)$.
\item 
$P_i$ sends the output shares $\left( \share{d_{0, j}}_i \right)_{j = 0}^n, \left( \share{d_{1, j}}_i \right)_{j = 0}^n$ to $P^*$.

\item
After receiving all shares, $P^*$ reconstructs for every $j \in \{0,\ldots,n\}$:
$$
d_{0, j} := \share{d_{0, j}}_0 + \share{d_{0, j}}_1 {\rm~~~~and~~~~} d_{1, j} := \share{d_{1, j}}_0 + \share{d_{1, j}}_1 \ ,
$$ 
and finally executes Algorithm~\ref{alg:final} to receive
$(b_0, b_1) := \mathsf{ComparisonFinal} \left( \left( d_{0, j} \right)_{j = 0}^n, \left( d_{1, j} \right)_{j = 0}^n \right)$.
\item $P^*$ sends $b_i$ to $P_i$. 
\end{enumerate}
\item {\bf Output:} Each $P_i$ outputs $b_i$. $P^*$ outputs $(b_0,b_1)$. 
\end{itemize}
\end{protocol}
\end{mdframed}
\end{figure*}

\begin{figure*}
\begin{mdframed}
\begin{protocol}[Maliciously secure comparison protocol $\picomp$] \label{prot:malicious}
\leavevmode
\noindent
\textbf{Input:} $P_0$ and $P_1$ hold integers $v_0$ and $v_1$, respectively, in $\{0, \ldots , 2^n - 1\}$. $P^*$ has no input.

\noindent
\textbf{Setup phase:} A coin-tossing protocol $\Pi_{\sf CT}$, 
and a 
commitment scheme $(\mathsf{Gen}, \mathsf{Com})$, are chosen (see Sect.~\ref{sec:preliminaries},\ref{appx:preliminaries}). 

\noindent
{\bf The protocol:} 
\begin{MyEnumerate}

\item Commit  $V_i \gets \mathsf{Com}(v_i; r_i)$, and send $V_i$ to $P^*$. $P^*$ delivers $V_{i}$ to $P_{1-i}$. 

\item Engage with $P_{1-i}$ in the coin-tossing procedure $\Pi_{\sf CT}$, and obtain a $\lambda$-bit shared $s$. 
\item Compute the bit decomposition $v_i = \sum_{j \in N} 2^{n - 1 - j} \cdot v_{i, j}$, for bits $v_{i, j} \in \{0, 1\}$. 
\item For each $j \in N$:\label{prot:malicious:generate-shares}
\begin{MyEnumerate}
\item  Compute a random additive secret-sharing $v_{i, j} = \share{v_{i, j}}_0 + \share{v_{i, j}}_1$ in $\mathbb{F}_q$. 
\item 
Commit $V_{i, j, k} \gets \mathsf{Com}(\share{v_{i, j}}_k; ~r_{i, j, k})$ for each $k \in \bit$. 
\item 
Open  $V_{i, j, 1 - i}$ by sending $\share{v_{i, j}}_{1 - i}$ and $r_{i, j, 1 - i}$ to $P_{1 - i}$. 
\end{MyEnumerate}
\item Send the full array $\left( V_{i, j, k} \right)_{j, k = 0}^{n - 1, 1}$ to $P_{1 - i}$. \label{prot:malicious:send-array}
\item Compute: 
$\pi_i \gets \mathsf{ComEq.Prove}\left( V_i, \prod_{j = 0}^{n - 1} \left( V_{i, j, 0} \cdot V_{i, j, 1} \right)^{2^{n - 1 - j}} \right),$\label{stp:malicious:comp-eq}
$\pi_{i,j} \gets \mathsf{BitProof.Prove} \left(V_{i, j, 0} \cdot V_{i, j, 1} \right)$, for all $j \in N$,\label{stp:malicious:bit-proof}

%
$P_i$ sends $\pi_i$, and $(\pi_{i,j})_{j = 0}^{n - 1}$ to $P_{1 - i}$.
\item $P_i$ receives  $\pi_{1-i}$, and $(\pi_{1-i,j})_{j = 0}^{n - 1}$ and verifies the following:
\begin{MyEnumerate}
\item The openings $\share{v_{1 - i, j}}_i$ and $r_{1 - i, j, i}$ indeed open $V_{1 - i, j, i}$, for $j \in N$,
\item $\mathsf{ComEq.Verify} \left( \pi_{1-i}, V_{1 - i}, \prod_{j = 0}^{n - 1} \left( V_{1 - i, j, 0} \cdot V_{1 - i, j, 1} \right)^{2^{n - 1 - j}} \right)$,
\item $\mathsf{BitProof.Verify} \left(\pi_{1-i,j}, V_{1 - i, j, 0} \cdot V_{1 - i, j, 1} \right)$ for each $j \in N$.
\end{MyEnumerate}
If any of these checks fail, $P_i$ aborts.


\item $P_i$ runs Algorithm~\ref{alg:initial}, in parallel, on the shares $\share{v_{k, j}}_i$, the randomnesses $r_{k, j, i}$, and the commitments to the \textit{other party}'s shares $V_{k, j, 1 - i}$ (all for $j \in N$ and $k \in \{0, 1\}$). That is, $P_i$ runs:
\begin{gather*}
\left( \left( \share{d_{0, j}}_i \right)_{j = 0}^n, \left( \share{d_{1, j}}_i \right)_{j = 0}^n \right) \gets \mathsf{ComparisonInitial}\left( \left( \share{v_{0, j}}_i \right)_{j = 0}^{n - 1}, \left( \share{v_{1, j}}_i \right)_{j = 0}^{n - 1} \right), \\
\left( \left( s_{0, j, i} \right)_{j = 0}^n, \left( s_{1, j, i} \right)_{j = 0}^n \right) \gets \mathsf{ComparisonInitial}\left( \left( r_{0, j, i} \right)_{j = 0}^{n - 1}, \left( r_{1, j, i} \right)_{j = 0}^{n - 1} \right), \\
\left( \left( D_{0, j, 1 - i} \right)_{j = 0}^n, \left( D_{1, j, 1 - i} \right)_{j = 0}^n \right) \gets \mathsf{ComparisonInitial}\left( \left( V_{0, j, 1 - i} \right)_{j = 0}^{n - 1}, \left( V_{1, j, 1 - i} \right)_{j = 0}^{n - 1} \right),
\end{gather*}
using the same shared randomness $s$ for all internal coin flips.

\item $P_i$ sends $\left( \share{d_{0, j}}_i \right)_{j = 0}^n, \left( \share{d_{1, j}}_i \right)_{j = 0}^n$ and the randomnesses $\left( s_{0, j, i} \right)_{j = 0}^n, \left( s_{1, j, i} \right)_{j = 0}^n$ to $P^*$.\label{stp:malicious-send-to-P-star}

\end{MyEnumerate}

\noindent
{\bf \boldmath Party $P^*$ (Output reconstruction):} 
After receiving all shares, $P^*$ proceeds as follows:
\begin{MyEnumerate}
\item Reconstruct for each $j \in N$:
\vspace{-1ex}
\begin{eqnarray*}
d_{0, j} := \share{d_{0, j}}_0 + \share{d_{0, j}}_1 \ , &~~~~~& s_{0, j} := s_{0, j, 0} + s_{0, j, 1} \\
d_{1, j} := \share{d_{1, j}}_0 + \share{d_{1, j}}_1 \ , &~~~~~& s_{1, j} := s_{1, j, 0} + s_{1, j, 1}
\end{eqnarray*}
\item Finally, $P^*$ executes Algorithm~\ref{alg:final}, that is:
$
(b_0, b_1) := \mathsf{ComparisonFinal} \left( \left( d_{0, j} \right)_{j = 0}^n, \left( d_{1, j} \right)_{j = 0}^n \right) 
$
\item For each $i \in \{0, 1\}$, if $b_i$ is true, then $P^*$ re-commits $D_{i, j} := \mathsf{Com}(d_{i, j}; s_{i, j})$ for each $j \in N$, computes $\pi_i' \gets \textsf{OneMany.Prove}\left( \left( D_{i, j} \right)_{j = 0}^n \right)$, and finally sends $\pi_i'$ to $P_i$.  Otherwise, $P^*$ sends $\bot$ to $P_i$. $P^*$ outputs both $b_0$ and $b_1$. 
\end{MyEnumerate}

\smallskip
\noindent
{\bf \boldmath Each Party $P_i$ (output reconstruction):} 
\begin{MyEnumerate}
\item If receives a proof $\pi_i'$, compute $D_{i, j} := \mathsf{Com}(\share{d_{i, j}}_i; s_{i, j, i}) \cdot D_{i, j, 1 - i}$, for each $j \in N$, and then verifies $\textsf{OneMany.Verify}\left( \pi_i', \left( D_{i, j} \right)_{j = 0}^n \right)$. If verification passes then output {\sf true}, otherwise, output {\sf false}. 
\end{MyEnumerate}

\end{protocol}
\end{mdframed}

\end{figure*}

\begin{algorithm}[H]
\caption{$\textsf{ComparisonFinal}\left( \left(d_{0, 0}, \ldots , d_{0, n}\right),\right.\\ \left. \hspace{+0.6\columnwidth}\left(d_{1, 0}, \ldots , d_{1, n}\right)\right)$}
\label{alg:final}
\begin{algorithmic}[1]
\State Assign $b_0 := \textbf{any}_{j = 0}^n \left( d_{0, j} = 0 \right)$ 
\State Assign $b_1 := \textbf{any}_{j = 0}^n \left( d_{1, j} = 0 \right)$
\State \Return $b_0$ and $b_1$
\end{algorithmic}
\end{algorithm}

In the below theorem, we again consider bit-decomposed integers $v_0 = \sum_{j = 0}^{n - 1} 2^{n - 1 - j} \cdot v_{0, j}$ and $v_1 = \sum_{j = 0}^{n - 1} 2^{n - 1 - j} \cdot v_{1, j}$; we view the bits $v_{i, j}$ as elements of $\{0, 1\} \subset \mathbb{F}_q$. The following theorem is proven in Appendix~\ref{appx:thm:correctness}:

\begin{theorem} \label{thm:correctness}
Suppose $n$ is such that $2 + 4 \cdot (2^n - 1) < q$. Then for every $v_0,v_1 \in \mathbb{F}_q$\vspace{-1ex}
\begin{eqnarray*}
\lefteqn{(v_0 \leq v_1, v_1\leq v_0) = }\\
&& \mathsf{ComparisonFinal}\left(\mathsf{ComparisonInitial}\left(\vec{v_0},\vec{v_1}\right)\right) \ ,
\end{eqnarray*}
where $\vec{v_0},\vec{v_1}$ are the bit-decomposition of $v_0,v_1$, respectively. 
Moreover, for every $i\in \bit$:
\begin{MyItemize}
\item If $v_i \leq v_{1-i}$ then there exists exactly one $j \in \{0,\ldots,n\}$ such that $d_{i,j}=0$. Moreover, $j$ is distributed uniformly in $\{0,\ldots,n\}$, and each $d_{i,k}$ for $k \neq j$ is distributed uniformly in $\mathbb{F}_q\setminus \{0\}$. 
\item If $v_0 >v_{1-i}$ then the vector $(d_{i,0,},\ldots,d_{i,n})$ is distributed uniformly in $\mathbb{F}_q^{n+1}\setminus \{0\}^{n+1}$. 
\end{MyItemize}
\end{theorem}

We emphasize that Algorithm~\ref{alg:initial} uses only $\mathbb{F}_q$-linear operations throughout. A number of our below protocols conduct Algorithm~\ref{alg:initial} ``homomorphically''; that is, they execute the algorithm on elements of an $\mathbb{F}_q$-module $M$ which is unequal to $\mathbb{F}_q$ itself. As a basic example, Algorithm~\ref{alg:initial} may be executed on bits $(v_{0, j})_{j = 0}^{n - 1}$ and $(v_{1, j})_{j = 0}^{n - 1}$ which are \textit{committed}, provided that the commitment scheme is homomorphic (its message, randomness and commitment spaces should be $\mathbb{F}_q$-modules, and its commitment function an $\mathbb{F}_q$-module homomorphism). Furthermore, Algorithm~\ref{alg:initial} may be conducted on additive $\mathbb{F}_q$-shares of the bits $(v_{0, j})_{j = 0}^{n - 1}$ and $(v_{1, j})_{j = 0}^{n - 1}$.

In this latter setting, sense must be given to the affine additive constants $\pm 1$. As in \cite{Wagh:2019aa}, we specify that these be shared in the obvious way; that is, we stipulate that $P_0$ and $P_1$ use the shares $0$ and $\pm 1$, respectively.

\subsection{The Semi-Honest Protocol}\label{sec:minimum:semi-honest}

For simplicity, we first describe a protocol that securely computes this functionality in the setting of three-party computation with an honest majority and a \textit{semi-honest} adversary. We give a maliciously secure version in Protocol \ref{prot:malicious} and prove its security in Section \ref{appx:thm:semi} of the Appendix.

\begin{theorem}
\label{thm:semi}
If ~$\Pi_{\sf CT}$ is a secure coin-tossing protocol, $G$ is a pseudorandom generator, and the two clients communicate using symmetric authenticated encryption with pseudorandom ciphertexts, then Protocol~\ref{prot:semi-honest} securely computes Functionality~\ref{func:compare}   in the presence of a semi-honest adversary corrupting at most one party. Each party sends or receives $O(n^2+\lambda)$ bits, where $n$ is the length of the input and $\lambda$ is the security parameter. 
\end{theorem}



\subsection{The Maliciously Secure Protocol}
\label{sec:minimum:malicious}
We now give our malicious protocol for Functionality~\ref{func:compare} in Protocol~\ref{prot:malicious}.
To ease notation, we denote $N=\{0,\ldots,n-1\}$. 
We already gave an overview of the protocol as part of the introduction. The parties commit to the inputs, share their inputs, commit to the shares, and prove that all of those are consistent. Then, each party can operate on the shares it received (as in the semi-honest protocol), but also on the commitment that it holds on the other parties' share. When $P^*$ receives the shares of $d$ from party $P_i$, it also receives a commitment of what $P_{1-i}$ is supposed to send. The server can therefore check for consistency, and that no party cheated. Moreover, each party $P_i$ can also compute commitments of the two vectors $\vec{d_0}, \vec{d_1}$. When the server comes to prove $P_i$ that $v_i \leq v_{1-i}$, it has to show that there is a $0$ coordinate in its vector $\vec{d_i}$. This is possible using one-out-of-many proof. See Appendix~\ref{sec:malicious-proof}  for the full security proof. We note that the functionality is slightly different than that of Functionality~\ref{func:compare}.



\subsection{Bank-to-Client}
We show in Appendix~\ref{sec:twoparty} a two-party (Bank-to-Client) version of the protocol, where the bank does not just facilitate the computation but also provides input. In a nutshell, the parties use linear homomorphic encryption -- ElGamal encryption -- instead of secret sharing.


 \begin{table*}
  \begin{center}
  \begin{footnotesize}
    \begin{tabular}[t]{c|c|c|c|c|cc}
      \toprule
       &  & &   Throughput  & Matching  & Matching \\
       & Number of    & & (transactions/    &  received msg  & sent msg \\

     & Symbols  &Latency (sec)  & sec)         & Size (MB)  &Size (MB) \\
      \midrule
     \multirow{7}{*}{Bank-to-client} 
     & 100  &9.903 ($\pm$0.174 ) &10.09 & 0.215&0.521\\

      & 200       &19.533 ($\pm$0.530 ) &10.23 & 0.430&1.040\\
  


      & 500  & 46.223($\pm$0.787 ) &10.81 & 1.076&2.597\\

      & 1000       &95.396 ($\pm$2.063 ) &10.48& 2.152&5.194\\

      & 2000 &183.186 ($\pm$2.512 )  &10.91   &  4.304&10.382\\
  
 
      & 4000  &356.740 ($\pm$2.149 ) &11.21 & 8.608&20.762\\

  
         & 10000 &941.813 ($\pm$18.465 )  &10.61   &  21.520&51.902\\\midrule
        \multirow{7}{*}{Client-to-client} &

        100  &11.15 ($\pm$0.060 ) &8.96 & 0.972&1.549\\

      & 200       &20.636 ($\pm$0.525 ) &9.69 & 1.945&3.096\\
  


      & 500  & 51.493($\pm$2.343 ) &9.71 & 4.863&7.737\\

      & 1000       &101.051 ($\pm$2.587 ) &9.89 & 9.727&15.472\\

      & 2000 &208.813 ($\pm$1.479 )  & 9.57  &  19.454&30.942\\
  
 
      & 4000  &390.510 ($\pm$3.020 ) & 10.24& 38.908&61.882\\

  
        & 10000 &1064.443 ($\pm$70.439 )  & 9.39  &  97.270&154.702\\

      \bottomrule
    \end{tabular}
  \end{footnotesize}
  \end{center}
  \vspace{-2ex}
  \caption{Performance of Bank-to-Client matching for two clients, and Client-to-Client matching.}
  \label{fig:2partyAWS}
  \label{fig:3partyAWS}
\end{table*}%

\section{Prime Match System Performance}
\label{sec:prime-match-perf}
We report benchmarks of Prime Match in two different environments, a Proof of Concept (POC) environment, and the production environment after refactoring the code to meet the requirements of the bank's systems. The former benchmarks can be used to value the performance of the comparison protocol for other applications in different systems. 

\paragraph{Secure Minimum Protocol Performance}\label{sec:exp1}: For the purposes of practical convenience, adoption, and portability, our client module is entirely browser-based, written in JavaScript. Its cryptographically intensive components are written in the C language with side-channel resistance, compiled using Emscripten into WebAssembly (which also runs natively in the browser). Our server is written in Python, and also executes its cryptographically intensive code in C. Both components are multi-threaded—using WebWorkers on the client side and a thread pool on the server’s—and can execute arbitrarily many concurrent instances of the protocol in parallel (i.e., constrained only by hardware). All players communicate by sending binary data on WebSockets (all commitments, proofs, and messages are serialized). 


We run our experiments on commodity hardware throughout since our implementation is targeted to a real-world application where clients hold conventional computers. In particular, one of the two clients runs on an Intel Core i7 processor, with 6 cores, each 2.6GHz, and another one runs on an Intel Core i5, with 4 cores, each 2.00 GHz. Both of them are Windows machines. Our server runs in a Linux AWS instance of type c5a.8xlarge, with 32 vCPUs. In the first scenario, we run the client-to-bank inventory matching protocol for two clients where we process each client one by one against the bank's inventory. In Table~\ref{fig:2partyAWS} we report the performance for a different number of registered orders per client $(100,200,\ldots,10000)$. \textit{Latency} refers to the total time it takes to process all the orders from both clients (in seconds). \textit{Throughput} measures the number of transactions per second. The number of orders/symbols processed per second is approximately $10$. We also report the message size (in MB) for the message sent from each client to the server during the registration phase and the matching phase. We further record the size of the messages received from the server to each client. The bandwidth is 300 Mbps. Note that the message sizes can be reduced considerably if message serialization is not required.

In the second scenario, we run the client-to-client inventory matching protocol (the comparison protocol) for two clients where we try to match the orders of the two clients via the bank. In this case we assume that the bank has no inventory.
We report the performance in Table~\ref{fig:3partyAWS}. The number of orders/symbols processed per second is approximately $10$ in this case too.

\paragraph{Prime Match Performance in Production}\label{sec:exp2}:
Figure~\ref{fig:arch} shows a sketch of both the bank's network tiering and the application architecture of Prime Match. The right side of the diagram shows that the clients are able to access the Prime Match UI through the bank's Markets portal with the correct entitlement.


%


		The application UI for Prime Match is hosted on the bank's internal cloud platform. 
		After each client is getting authenticated by the bank's Markets portal, it can access the application on the browser 
		which establishes web socket connections to the server. The client's computation takes place locally within their browser, ensuring that all private data remains local. The server is hosted on the bank's trade management platform.

%
%
%
%
%
%

The bank's network tiering exists to designate the network topology in and between approved security gateways (firewalls). The network traffic from the client application is handled by tier 1 which helps the bank's Internet customers achieve low latency, secured, and accelerated content access to internally hosted applications. Then, the traffic will be subsequently forwarded to a web socket tunnel in tier 2 and gets further directed to the server in tier 3.



 During the axe registration phase, the client logins to Prime Match from his/her web browser and uploads a file of orders (symbol, directions, quantity) which are encrypted locally and securely on the browser. The client is uploading a minimum (threshold) and a maximum (full) quantity per symbol. Prime Match first processes the encrypted threshold quantities of all clients and then it processes the full quantities. This process implements the range functionality of Appendix~\ref{sec:range}. Note that for the partial matches, the full quantity is never revealed to the server. Moreover, if there is no match no quantity is revealed to the server. Figure~\ref{fig:finalmatch} shows a complete run of an auction on the client's browser after the matching phase where the final matched quantities of some test symbols/securities with the corresponding fixed spread are revealed.

   Our protocol runs in production every 30 minutes. There are two match runs each hour and matching starts at xx:10 and xx:40. Axe registration starts 8 minutes before matching (xx:02 and xx:32). The matching process finishes at various times (as shown in the experiments section) based on the number of symbols. Based on the application requirements (up to 5000 symbols in the US and 10-60 clients) it does not finish after xx:55.  
 
For running the code in production in the bank's environment, the two clients run on the same type of Windows machines whose specs are Intel XEON CPU E3-1585L, with 4 cores, each 3.00 GHz. The server runs on the bank's trade management platform with 32 vCPUs. In Table~\ref{fig:primematchperf} of the Appendix, we report the performance. Moreover, we discuss challenges we faced during the implementation in Appendix~\ref{app:imp}.




\section{Conclusion}

Inventory matching is a fundamental service in the traditional financial world. In this work, we introduce secure multiparty computation in financial services by presenting a solution for matching orders in a stock exchange while maintaining the privacy of the orders. Information is revealed only if there is a match. Our central tool is a new protocol for secure comparison with linear operations in the presence of a malicious adversary, which can be of independent interest. Our system is running live, in production, and is adopted by a large bank in the US -- J.P. Morgan.

 \section{Acknowledgements}

We would like to thank Mike Reich, Vaibhav Popat, Sitaraman Rajamani, Dan Stora, James Mcilveen, Oluwatoyin Aguiyi, Niall Campbell, Wanyi Jiang, Grant McKenzie, Steven Price, Vinay Gayakwad, Srikanth Veluvolu, Noel Peters for their great efforts and help to move Prime Match in production. Last but not least, we would like to thank our executive sponsor Jason Sippel. The paper is based on joint patents~\cite{PA,PB}

This paper was prepared in part for information purposes by the Artificial Intelligence Research group and AlgoCRYPT CoE of JPMorgan Chase \& Co and its affiliates (``JP Morgan''), and is not a product of the Research Department of JP Morgan. JP Morgan makes no representation and warranty whatsoever and disclaims all liability, for the completeness, accuracy or reliability of the information contained herein. This document is not intended as investment research or investment advice, or a recommendation, offer or solicitation for the purchase or sale of any security, financial instrument, financial product or service, or to be used in any way for evaluating the merits of participating in any transaction, and shall not constitute a solicitation under any jurisdiction or to any person, if such solicitation under such jurisdiction or to such person would be unlawful. 2023 JP Morgan Chase \& Co. All rights reserved.

\bibliographystyle{plain}
\bibliography{mainUSENIX}

\begin{thebibliography}{10}

\bibitem{Sha256circuit}
{'Bristol Fashion' MPC Circuits}.
\newblock \url{https://homes.esat.kuleuven.be/~nsmart/MPC/}, accessed: December
  2022.

\bibitem{cryptoeprint:2022/866}
Amit Agarwal, Stanislav Peceny, Mariana Raykova, Phillipp Schoppmann, and Karn
  Seth.
\newblock Communication efficient secure logistic regression.
\newblock Cryptology ePrint Archive, Paper 2022/866, 2022.
\newblock \url{https://eprint.iacr.org/2022/866}.

\bibitem{ArakiFLNO16}
Toshinori Araki, Jun Furukawa, Yehuda Lindell, Ariel Nof, and Kazuma Ohara.
\newblock High-throughput semi-honest secure three-party computation with an
  honest majority.
\newblock In {\em Proceedings of the 2016 {ACM} {SIGSAC} Conference on Computer
  and Communications Security, Vienna, Austria, October 24-28, 2016}, pages
  805--817. {ACM}, 2016.

\bibitem{PA}
Gilad Asharov, Tucker Balch, Hans Buehler, Richard Hua, Antigoni
  Polychroniadou, and Manuela Veloso.
\newblock Systems and methods for privacy-preserving inventory matching, filed
  Dec.~6, 2019.

\bibitem{Asharov:2020aa}
Gilad Asharov, Tucker~Hybinette Balch, Antigoni Polychroniadou, and Manuela
  Veloso.
\newblock {PPDPs}: Privacy-preserving dark pools.
\newblock In {\em 19th International Conference on Autonomous Agents and
  Multi-Agent Systems (AAMAS 2020)}, 2020.
\newblock Extended abstract.

\bibitem{bag2019seal}
Samiran Bag, Feng Hao, Siamak~F Shahandashti, and Indranil~Ghosh Ray.
\newblock Seal: Sealed-bid auction without auctioneers.
\newblock {\em IEEE Transactions on Information Forensics and Security},
  15:2042--2052, 2019.

\bibitem{secretmatch}
Tucker Balch, Benjamin~E. Diamond, and Antigoni Polychroniadou.
\newblock Secretmatch: Inventory matching from fully homomorphic encryption.
\newblock In {\em Proceedings of the First ACM International Conference on AI
  in Finance}, ICAIF '20, New York, NY, USA, 2020. Association for Computing
  Machinery.

\bibitem{BarakCLPR11}
Boaz Barak, Ran Canetti, Yehuda Lindell, Rafael Pass, and Tal Rabin.
\newblock Secure computation without authentication.
\newblock {\em J. Cryptol.}, 24(4):720--760, 2011.

\bibitem{BogetoftCDGJKNNNPST09}
Peter Bogetoft, Dan~Lund Christensen, Ivan Damg{\aa}rd, Martin Geisler,
  Thomas~P. Jakobsen, Mikkel Kr{\o}igaard, Janus~Dam Nielsen, Jesper~Buus
  Nielsen, Kurt Nielsen, Jakob Pagter, Michael~I. Schwartzbach, and Tomas Toft.
\newblock Secure multiparty computation goes live.
\newblock In {\em Financial Cryptography and Data Security, {FC} 2009}, volume
  5628, pages 325--343. Springer, 2009.

\bibitem{Bunz:2020aa}
Benedikt B\"{u}nz, Shashank Agrawal, Mahdi Zamani, and Dan Boneh.
\newblock Zether: Towards privacy in a smart contract world.
\newblock In {\em International Conference on Financial Cryptography and Data
  Security}, 2020.
\newblock Full version.

\bibitem{Bunz:2018aa}
Benedikt B\"{u}nz, Jonathan Bootle, Dan Boneh, Andrew Poelstra, Pieter Wuille,
  and Greg Maxwell.
\newblock Bulletproofs: Short proofs for confidential transactions and more.
\newblock In {\em 2018 IEEE Symposium on Security and Privacy (SP)}, volume~1,
  2018.
\newblock Full version.

\bibitem{CampanelliGGN17}
Matteo Campanelli, Rosario Gennaro, Steven Goldfeder, and Luca Nizzardo.
\newblock Zero-knowledge contingent payments revisited: Attacks and payments
  for services.
\newblock In {\em Proceedings of the 2017 {ACM} {SIGSAC} Conference on Computer
  and Communications Security, {CCS} 2017}, pages 229--243. {ACM}, 2017.

\bibitem{Canetti00}
Ran Canetti.
\newblock Security and composition of multiparty cryptographic protocols.
\newblock {\em J. Cryptology}, 13(1):143--202, 2000.

\bibitem{Cartlidge:2019aa}
John Cartlidge, Nigel~P. Smart, and Younes~Talibi Alaoui.
\newblock Mpc joins the dark side.
\newblock In {\em Proceedings of the 2019 ACM Asia Conference on Computer and
  Communications Security}, pages 148--159. Association for Computing
  Machinery, 2019.
\newblock Full version.

\bibitem{cryptoeprint:2020:662}
John Cartlidge, Nigel~P. Smart, and Younes~Talibi Alaoui.
\newblock Multi-party computation mechanism for anonymous equity block trading:
  A secure implementation of turquoise plato uncross.
\newblock Cryptology ePrint Archive, Report 2020/662, 2020.
\newblock \url{https://ia.cr/2020/662}.

\bibitem{Catrina2010ImprovedPF}
Octavian Catrina and Sebastiaan de~Hoogh.
\newblock Improved primitives for secure multiparty integer computation.
\newblock In {\em SCN}, 2010.

\bibitem{ChidaGHIKLN18}
Koji Chida, Daniel Genkin, Koki Hamada, Dai Ikarashi, Ryo Kikuchi, Yehuda
  Lindell, and Ariel Nof.
\newblock Fast large-scale honest-majority {MPC} for malicious adversaries.
\newblock In {\em Advances in Cryptology - {CRYPTO} 2018}, volume 10993, pages
  34--64. Springer, 2018.

\bibitem{FC}
Mariana~Botelho da~Gama, John Cartlidge, Antigoni Polychroniadou, Nigel~P.
  Smart, and Younes~Talibi Alaoui.
\newblock Kicking-the-bucket: Fast privacy-preserving trading using buckets.
\newblock Financial Cryptography, 2022.

\bibitem{DBLP:conf/tcc/DamgardFKNT06}
Ivan Damg{\aa}rd, Matthias Fitzi, Eike Kiltz, Jesper~Buus Nielsen, and Tomas
  Toft.
\newblock Unconditionally secure constant-rounds multi-party computation for
  equality, comparison, bits and exponentiation.
\newblock In {\em Theory of Cryptography, {TCC} 2006}, volume 3876, pages
  285--304. Springer, 2006.

\bibitem{DamgardGK08}
Ivan Damg{\aa}rd, Martin Geisler, and Mikkel Kr{\o}igaard.
\newblock Homomorphic encryption and secure comparison.
\newblock {\em Int. J. Appl. Cryptogr.}, 1(1):22--31, 2008.

\bibitem{DamgardGK09}
Ivan Damg{\aa}rd, Martin Geisler, and Mikkel Kr{\o}igaard.
\newblock A correction to 'efficient and secure comparison for on-line
  auctions'.
\newblock {\em Int. J. Appl. Cryptogr.}, 1(4):323--324, 2009.

\bibitem{PB}
Benjamin Diamond and Antigoni Polychroniadou.
\newblock Privacy-preserving inventory matching with security against malicious
  adversaries, filed Jan.~28, 2021.

\bibitem{Fischlin01}
Marc Fischlin.
\newblock A cost-effective pay-per-multiplication comparison method for
  millionaires.
\newblock In {\em {CT-RSA} 2001}, volume 2020, pages 457--472. Springer, 2001.

\bibitem{galal2021publicly}
Hisham~S Galal and Amr~M Youssef.
\newblock Publicly verifiable and secrecy preserving periodic auctions.
\newblock In {\em International Conference on Financial Cryptography and Data
  Security}, pages 348--363. Springer, 2021.

\bibitem{Groth:2015aa}
Jens Groth and Markulf Kohlweiss.
\newblock One-out-of-many proofs: Or how to leak a secret and spend a coin.
\newblock In {\em Advances in Cryptology -- {EUROCRYPT} 2015}, volume 9057 of
  {\em Lecture Notes in Computer Science}, pages 253--280. Springer Berlin
  Heidelberg, 2015.

\bibitem{KatzRR018}
Jonathan Katz, Samuel Ranellucci, Mike Rosulek, and Xiao Wang.
\newblock Optimizing authenticated garbling for faster secure two-party
  computation.
\newblock In {\em Advances in Cryptology - {CRYPTO} 2018 - 38th Annual
  International Cryptology Conference, Santa Barbara, CA, USA, August 19-23,
  2018, Proceedings, Part {III}}, volume 10993 of {\em Lecture Notes in
  Computer Science}, pages 365--391. Springer, 2018.

\bibitem{LinT05}
Hsiao{-}Ying Lin and Wen{-}Guey Tzeng.
\newblock An efficient solution to the millionaires' problem based on
  homomorphic encryption.
\newblock In {\em Applied Cryptography and Network Security, Third
  International Conference, {ACNS} 2005, New York, NY, USA, June 7-10, 2005,
  Proceedings}, volume 3531 of {\em Lecture Notes in Computer Science}, pages
  456--466, 2005.

\bibitem{DBLP:conf/icalp/LipmaaT13}
Helger Lipmaa and Tomas Toft.
\newblock Secure equality and greater-than tests with sublinear online
  complexity.
\newblock In {\em Automata, Languages, and Programming - 40th International
  Colloquium, {ICALP} 2013, Riga, Latvia, July 8-12, 2013, Proceedings, Part
  {II}}, volume 7966 of {\em Lecture Notes in Computer Science}, pages
  645--656. Springer, 2013.

\bibitem{DBLP:conf/fc/MakriRVW21}
Eleftheria Makri, Dragos Rotaru, Frederik Vercauteren, and Sameer Wagh.
\newblock Rabbit: Efficient comparison for secure multi-party computation.
\newblock In {\em Financial Cryptography and Data Security - {FC} 2021}, volume
  12674, pages 249--270. Springer, 2021.

\bibitem{MassacciNN0W18}
Fabio Massacci, Chan~Nam Ngo, Jing Nie, Daniele Venturi, and Julian Williams.
\newblock Futuresmex: Secure, distributed futures market exchange.
\newblock In {\em 2018 {IEEE} Symposium on Security and Privacy, {SP} 2018,
  Proceedings, 21-23 May 2018, San Francisco, California, {USA}}, pages
  335--353. {IEEE} Computer Society, 2018.

\bibitem{MazloomDPB23}
Sahar Mazloom, Benjamin~E. Diamond, Antigoni Polychroniadou, and Tucker Balch.
\newblock An efficient data-independent priority queue and its application to
  dark pools.
\newblock {\em Proc. Priv. Enhancing Technol.}, 2023(2):5--22, 2023.

\bibitem{ngo2021practical}
Chan~Nam Ngo, Fabio Massacci, Florian Kerschbaum, and Julian Williams.
\newblock Practical witness-key-agreement for blockchain-based dark pools
  financial trading.

\bibitem{DBLP:journals/ieicet/NishideO07}
Takashi Nishide and Kazuo Ohta.
\newblock Constant-round multiparty computation for interval test, equality
  test, and comparison.
\newblock {\em {IEICE} Trans. Fundam. Electron. Commun. Comput. Sci.},
  90-A(5):960--968, 2007.

\bibitem{PolychroniadouA23}
Antigoni Polychroniadou, Gilad Asharov, Benjamin~E. Diamond, Tucker Balch, Hans
  Buehler, Richard Hua, Suwen Gu, Greg Gimler, and Manuela Veloso.
\newblock Prime match: {A} privacy-preserving inventory matching system.
\newblock In Joseph~A. Calandrino and Carmela Troncoso, editors, {\em 32nd
  {USENIX} Security Symposium, {USENIX} Security 2023, Anaheim, CA, USA, August
  9-11, 2023}. {USENIX} Association, 2023.

\bibitem{Terelius:2012aa}
Bj{\"o}rn Terelius and Douglas Wikstr{\"o}m.
\newblock Efficiency limitations of Σ-protocols for group homomorphisms
  revisited.
\newblock In {\em Security and Cryptography for Networks}, pages 461--476,
  Berlin, Heidelberg, 2012. Springer Berlin Heidelberg.

\bibitem{Wagh:2019aa}
Sameer Wagh, Divya Gupta, and Nishanth Chandran.
\newblock {SecureNN}: 3-party secure computation for neural network training.
\newblock In {\em Proceedings on Privacy Enhancing Technologies}, volume 2019,
  pages 26--49, 2019.

\bibitem{WangRK17}
Xiao Wang, Samuel Ranellucci, and Jonathan Katz.
\newblock Authenticated garbling and efficient maliciously secure two-party
  computation.
\newblock In {\em Proceedings of the 2017 {ACM} {SIGSAC} Conference on Computer
  and Communications Security, {CCS} 2017, Dallas, TX, USA, October 30 -
  November 03, 2017}, pages 21--37. {ACM}, 2017.

\end{thebibliography}


\appendix

\addcontentsline{toc}{section}{Appendices}

\section{Deferred Preliminaries}
\label{appx:preliminaries}

%
%


\paragraph{Notations.}
A distribution ensemble $X = \{X(a,\lambda)\}_{a \in {\cal D},\lambda \in \mathbb{N}}$ is an infinite sequence of random variables indexed by $a \in {\cal D}$ and $\lambda \in \mathbb{N}$. We let $x \gets X(a,\lambda)$ denote sampling an element $x$ according to the distribution $X(a,\lambda)$. Two distribution ensembles $X = \{X(a,\lambda)\}_{a \in {\cal D},\lambda \in \mathbb{N}}$ and $Y = \{Y(a,\lambda)\}_{a \in {\cal D},\lambda \in \mathbb{N}}$ are {\sf computationally-indistinguishable}, denoted $X \approx_{comp} Y$, if for every non-uniform {\sc ppt} algorithm $D$ there exists a negligible function ${\sf negl}(\cdot)$ such that for every $\lambda$ and $a \in {\cal D}$:
$$
\Pr\left[D(X(a,\lambda))=1\right]-\Pr\left[D(Y(a,\lambda))=1\right] \leq {\sf negl}(\lambda) \ .
$$

We let $\mathcal{G}$ denote a \textit{group-generation algorithm}, which on input $1^\lambda$ outputs a cyclic group $\mathbb{G}$, its prime order $q$ (with bit-length $\lambda$) and a generator $g \in \mathbb{G}$. We recall the {\sf discrete-logarithm assumption}, in which given a cyclic group $(\mathbb{G}, q, g)$ and a random $h\in \mathcal{G}$, any probabilistic polynomial time adversary finds $x$ such that $g^x=h$ with at most negligible probability in $\lambda$. The {\sf Decisional Diffie-Hellman assumption} states that any probabilistic polynomial time adversary distinguishes between $(\mathbb{G}, q, g, g^x, g^y,g^{xy})$ and $(\mathbb{G}, q, g, g^x, g^y,g^{z})$ for random $x,y,z \in \mathbb{F}_q$ with at most negligible probability in~$\lambda$. 

\paragraph{Secure multiparty computation.}
We use standard definitions of secure multiparty computation, following the stand-alone model~\cite{Canetti00}. We distinguish between security in the presence of a {\sf semi-honest adversary}, i.e., the adversary follows the protocol execution but looks at the all messages it received and tries to learn some additional information about the honest parties' inputs, and a {\sf malicious} adversary, i.e., an adversary that might act arbitrarily. 
We compare between a function $\cF(x_1,\ldots,x_n) = (y_1,\ldots,y_n)$ and a protocol $\pi(x_1,\ldots,x_n)$ that allegedly privately computes the function $\cF$.

We review some notations: 
Let $\Adv$ be a non-uniform probabilistic polynomial-time adversary controlling
parties $I \subset [n]$. Let $\real_{\pi,\Adv(z),I}(x_1,\ldots,x_n, \lambda)$ denote the output of the
honest parties and $\Adv$ in an real execution of $\pi$, with inputs $x_1,\ldots,x_n$, auxiliary input $z$ for $\Adv$, and security parameter $\lambda$. Let $\Sim$ be a non-uniform probabilistic polynomial-time adversary. Let $\ideal_{\cF,\Sim(z),I}(x_1, \ldots,x_n, \lambda)$ denote the output of
the honest parties and $\Sim$ in an ideal execution with the functionality $\cF$, inputs
$x_1, \ldots,x_n$ to the parties, auxiliary-input $z$ to $\Sim$, and security parameter~$\lambda$. The security definition requires that for every real-world adversary $\Adv$, there exists an ideal world adversary $\Sim$ such that $\real$ and $\ideal$ are indistinguishable. We refer to~\cite{Canetti00} for the full definition.

\paragraph{Hybrid model and composition.}
The hybrid model combines both the real and ideal models. Specifically, an execution of a protocol $\pi_\cF$ in the $\cG$-hybrid model, for some functionality $\cG$, involves the parties sending normal messages to each other (as in the real model), but they also have access to a trusted party computing $\cG$. The composition theorem of~\cite{Canetti00} implies that if $\pi_\cF$ securely computes some functionality $\cF$ in the $\cG$-hybrid model and a protocol $\pi_\cG$ computing $\cG$ then the protocol $\pi_\cF^\cG$, where every ideal call of $\cG$ is replaced with an execution of $\pi_\cG$ securely-computes $\cF$ in the real world.

\paragraph{On modeling security.}
We sometimes wish to prove that the protocol is secure with, e.g., malicious $P_0$ but semi-honest $P_1$. We use the malicious security formalization where we quantify for adversaries that control $P_1$ that do follow the protocol specifications. Note that in the ideal world, the simulator receives the input of the corrupted parties, whereas, in the malicious setting, we usually extract the input it actually uses in the protocol. In the case of a semi-honest adversary, the simulator does not have to extract the input, and it can just send the input it received to the trusted party. 

\paragraph{Coin-tossing.}
We define the coin-tossing functionality, $\cF_{\sf CT}$. The functionality assumes two parties, $P_1$ and $P_2$:
\begin{MyEnumerate}
\item {\bf Input:} Each party inputs $1^\lambda$ and some parameter $\ell > \sec$. 
\item {\bf The functionality:} Sample a uniform string $s\leftarrow \bit^{\ell}$ and give the two parties $s$.
\end{MyEnumerate}
This functionality can be realized easily, assuming commitment schemes and pseudorandom generators (security with abort).

\paragraph{Commitment Schemes.}
 The {\sf hiding} property of the commitment scheme states that for any two messages $m_0,m_1$, any adversary distinguishes between $\mathsf{Com}(\mathsf{params}, m_0; r_0)$ and $ \mathsf{Com}(\mathsf{params}, m_1; r_1)$ with at most negligible probability in $\lambda$. The {\sf binding} property states that given $\mathsf{params} \gets \text{Gen}(1^\lambda)$, any adversary can find $(m_0,r_0), (m_1,r_1)$ such that $\mathsf{Com}(\mathsf{params}, m_0; r_0)=\mathsf{Com}(\mathsf{params}, m_1; r_1)$ with at most negligible probability in $\lambda$. 

For example, we have the Pedersen commitment scheme (see, e.g., \cite[Def. 6]{Bunz:2018aa}). For $\mathsf{params}$ consisting of a prime $q$, a cyclic group $\mathbb{G}$ of order $q$, and elements $g$ and $h$ of $\mathbb{G}$, the Pedersen scheme maps elements $m$ and $r$ of $\mathbb{F}_q$ to $\mathsf{Com}(m; r) := g^m h^r$. If the discrete logarithm problem is hard with respect to $\mathcal{G}$, then the Pedersen scheme is computationally binding.

\subsection{Zero Knowledge Proofs}\label{appendix:sec:zk}

We present definitions for zero-knowledge arguments of knowledge, closely following \cite{Groth:2015aa} and \cite{Bunz:2018aa}. 
We also review a number of \textit{particular} constructions.

We posit a triple of interactive, probabilistic polynomial time algorithms $\Pi = (\mathsf{Setup}, \mathsf{Prove}, \mathsf{Verify})$. We fix a polynomial-time-decidable ternary relation $\mathcal{R} \subset (\{0, 1\}^*)^3$; by definition, each common reference string $\sigma \gets \mathsf{Setup}(1^\lambda)$ yields an NP language $L_\sigma = \{ x \mid \exists w : (\sigma, x, w) \in \mathcal{R} \}$. We denote by $\mathsf{tr} \gets \left< \mathsf{Prove}(\sigma, x, w), \mathsf{Verify}(\sigma, x) \right>$ the (random) transcript of an interaction between $\mathsf{Prove}$ and $\mathsf{Verify}$ on auxiliary inputs $(\sigma, x, w)$ and $(\sigma, x)$ (respectively). 
Abusing notation, we occasionally write $b \gets \left< \mathsf{Prove}(\sigma, x, w), \mathsf{Verify}(\sigma, x) \right>$ for the single bit indicating the verifier's state upon completing the interaction (i.e., \textit{accept} or \textit{reject}). The zero-knowledge proof constructions we use (see below) all feature relations $\mathcal{R}$ whose \textit{statements} consist of commitments and whose \textit{witnesses} consist of openings to these commitments. To simplify notation---we adopt a somewhat non-standard notational scheme, whereby we omit the witness from each call to $\mathsf{Prove}$. 

\paragraph{The Zero-Knowledge Functionality $\cF_{\sf ZK}({\cal R})$:} 
The zero-knowledge functionality $\cF_{\sf ZK}({\cal R})$ for an NP relation ${\cal R}$ is defined as follows:
\begin{MyItemize}
\item {\bf Input:} The prover and the verifier hold the same setup string $\sigma$ and input $x$, and the prover holds $w$. 
\item {\bf Output:} If ${\cal R}(\sigma,x,w)=1$ then it sends ${\sf accept}$ to the verifier. Otherwise, it sends ${\sf reject}$. 
\end{MyItemize}

The \textit{particular} zero-knowledge proof constructions we use all feature relations $\mathcal{R}$ whose \textit{statements} consist of commitments and whose \textit{witnesses} consist of openings to these commitments. More precisely---for each $\mathcal{R}$ we consider below, and for each setup string $\sigma$---each element $x$ of $L_\sigma$ consists of one more commitment, while any valid witness $w$ to $x$'s membership consists exactly of openings (i.e., message and randomness) to some or all among these commitments.

In light of this fact---and to simplify notation---we adopt a somewhat non-standard notational scheme, whereby we omit the witness from each call to $\mathsf{Prove}$, leaving it implicit in the statement. For example, for a language $L_\sigma$ whose statements consist of single commitments, we write $\Pi.\mathsf{Prove}(V)$---for some commitment $V$---to mean $\Pi.\mathsf{Prove}(\sigma, V, (m, r))$, where $V = \mathsf{Com}(m; r)$ (and the commitment scheme $(\mathsf{Gen}, \mathsf{Com})$ is implicit in the setup string $\sigma$). The expression $\Pi.\mathsf{Verify}(\pi, V)$ of course retains its usual meaning. Informally, we view each commitment $V$ as a ``data structure'' whose internal ``fields'' $m$ and $r$ may, optionally (that is, for the prover), be populated with an opening to $V$. The call $\Pi.\mathsf{Prove}(V)$ acts on $V$ by accessing its internal fields.

Moreover, given commitments $V_0$ and $V_1$ whose internal openings are populated, we understand $V_0 \cdot V_1$ as a third commitment whose openings are \textit{also} populated (in the obvious way, i.e., by addition).


\paragraph{Commitment equality proof.}
A simple Schnorr proof can be used to show that two Pedersen commitments open to the same message (recall that in Pedersen's commitment, the commitment scheme is ${\sf Com}(m;r)=g^mh^r$ for two group elements $g,h$). We specialize $(\mathsf{Gen}, \mathsf{Com})$ to the Pedersen scheme, and set $\sigma$ to a Pedersen base $(g, h)$. For completeness, we outline the details. The relevant relation is as follows:\vspace{-1ex}

\begin{eqnarray*}
\lefteqn{\mathcal{R}_{\mathsf{ComEq}} = \left\{ \left( \sigma, V_0, V_1 \right) \mid \exists \left( m_0, r_0, m_1, r_1 \right)\right.} \\
&&\left.{\rm ~s.t.~} V_i = \mathsf{Com}(m_i; r_i) \forall i \in \{0, 1\} \wedge m_0 = m_1 \right\}.
\end{eqnarray*}

The relation $\mathcal{R}_{\mathsf{ComEq}}$ expresses knowledge of openings to two commitments whose messages are \textit{equal}. The protocol is trivial; it essentially runs a Schnorr protocol on the difference between $V_0$ and $V_1$. This protocol can be derived as a special case of the ``basic protocol'' of Terelius and Wikstr\"{o}m \cite[Prot.~1]{Terelius:2012aa}; it also appears implicitly in the techniques of \cite[Sec.~G]{Bunz:2020aa}.
The prover sends $K = h^k$ for a random $k\leftarrow \mathbb{F}_q$. The verifier chooses  a random $x \leftarrow \mathbb{F}_q$. The prover replies with $s:=(r_0-r_1)\cdot x+k$ and the verifier accepts iff $h^s = (V_0 \cdot V_1^{-1})^s \cdot K$. 


\paragraph{Bit proof.}
We now recall two important protocols from Groth, and Kohlweiss \cite{Groth:2015aa}. The first, \cite[Fig.~1]{Groth:2015aa}, demonstrates knowledge of a secret opening $(m, r)$ of a public commitment $V$ for which $m \in \{0, 1\}$. We use the protocol $\mathsf{BitProof}$ from~\cite{Groth:2015aa}. The bit-proof protocol concerns the following relation: 

\begin{eqnarray*}
\lefteqn{\hspace{-14ex}\mathcal{R}_{\mathsf{BitProof}} = 
\left\{ \left( \sigma, V \right) \mid \exists~  (m, r), {\rm~s.t.~} V = \mathsf{Com}(m; r)\right.}\\
&& \left. \wedge m \in \{0, 1\} \right\} \ .
\end{eqnarray*}

%

\paragraph{One-out-of-many proofs.}
The second protocol---namely, \cite[Fig.~2]{Groth:2015aa}---demonstrates knowledge of a secret \textit{element} $V_l$ of a public list $(V_0, \ldots , V_n)$---as well as a secret randomness $r$---for which $V_l = \mathsf{Com}(0; r)$. Though \cite{Groth:2015aa} uses the former protocol merely as a subroutine of the latter, we will have occasion to make \textit{independent} use of both protocols. For a fixed integer $n$, we have the relation:
\begin{eqnarray*}
\lefteqn{\hspace{-6ex}\mathcal{R}_{\mathsf{OneMany}} = \left\{ \left( \sigma, (V_0, \ldots , V_n)\right) \mid \exists l \in [0,\ldots,n], r,\right.}\\
&& \hspace{+4ex} \left. {\rm ~s.t.~} V_l = \mathsf{Com}(0; r) \right\} \ ,
\end{eqnarray*}
as well as the proof protocol $\mathsf{OneMany}$, which we use exactly as written in \cite[Fig.~2]{Groth:2015aa} without modification. 

%

In practice, we choose $n$ which is one short of a power of 2 (i.e., $n = 2^m - 1$ for some $m$). The reason for this unusual indexing scheme will become clear below.

  \section{Proof of Theorem~\ref{thm:correctness}}
\label{appx:thm:correctness}
\begin{theorem}[Theorem~\ref{thm:correctness}, restated]
Suppose $n$ is such that $2 + 4 \cdot (2^n - 1) < q$. Then for every $v_0,v_1 \in \mathbb{F}_q$
\begin{eqnarray*}
\lefteqn{(v_0 \leq v_1, v_1\leq v_0) = }\\
&& \mathsf{ComparisonFinal}\left(\mathsf{ComparisonInitial}\left(\vec{v_0},\vec{v_1}\right)\right) \ ,
\end{eqnarray*}
where $\vec{v_0},\vec{v_1}$ are the bit-decomposition of $v_0,v_1$, respectively. 
Moreover, for every $i\in \bit$:
\begin{MyItemize}
\item If $v_i \leq v_{1-i}$ then there exists exactly one $j \in \{0,\ldots,n\}$ such that $d_{i,j}=0$. Moreover, $j$ is distributed uniformly in $\{0,\ldots,n\}$, and each $d_{i,k}$ for $k \neq j$ is distributed uniformly in $\mathbb{F}_q\setminus \{0\}$. 
\item If $v_0 >v_{1-i}$ then the vector $(d_{i,0,},\ldots,d_{i,n})$ is distributed uniformly in $\mathbb{F}_q^{n+1}\setminus \{0\}^{n+1}$. 
\end{MyItemize}
\end{theorem}

\begin{proof}
We recall that Algorithms \ref{alg:initial} and \ref{alg:final} are correct so long as, for each $j \in \{0, \ldots , n\}$, the quantities $c_{0, j}$ and $c_{1, j}$ equal 0 if and only if, respectively, $v_{0, j} < v_{1, j}$ and $v_{0, j} > v_{1, j}$ (if $j < n$) and moreover it holds that $v_{0, k} = v_{1, k}$ for each $k < j$ .

To show the nontrivial ``only if'' implications, it suffices to show that, for each $j$, the quantities $c_{0, j}$ and $c_{1, j}$ are non-zero so long as \textit{either} some $w_k \neq 0$, 
for $k < j$, or else the appropriate bit inequality (i.e., either $v_{0, j} < v_{1, j}$ or $v_{0, j} > v_{1, j}$) fails to hold. The special case in which all $w_k$ are 0 is trivially settled; in this case, $c_{0, j}$ and $c_{1, j}$ are easily seen to be nonzero when $v_{0, j} \geq v_{1, j}$ and $v_{0, j} \leq v_{1, j}$, respectively (this is guaranteed so long as $-2$ and $2$ are unequal to $0$ in $\mathbb{F}_q$, true so long as $q > 2$). To show the claim, then, we fix an arbitrary $j \in \{0, \ldots, n\}$, and consider an arbitrary set of coefficients $(w_k)_{k < j}$ for which each $w_k \in \{-1, 0, 1\}$ and \textit{at least one $w_k$ is nonzero}. It suffices to demonstrate, in this setting, that the sum $\sum_{k < j} 2^{2 + k} \cdot w_k$ cannot attain any of the values $\{-2, -1, 0, -1, 2\}$.

We first show that this property holds ``over $\mathbb{Z}$''; that is, we interpret all quantities (including the input bits $v_{0, 0}, \ldots, v_{0, n - 1}, v_{1, 0}, \ldots, v_{1, n - 1}$, as well as all scalars) as integers. We denote by $k^*$ the smallest index $k \in \{0, \ldots , j - 1\}$ for which $w_{k^*} \neq 0$. Clearly the sum $\sum_{k = 0}^{k^*} 2^{2 + k} \cdot w_k \not \in \{-2, -1, 0, -1, 2\}$; in fact, this sum equals either positive or negative $2^{2 + k^*}$ (whose absolute value is at least $4$). In particular, this sum's residue class modulo $2^{3 + k^*}$ is $2^{2 + k^*}$.

We now argue that $\sum_{k = 0}^{j - 1} 2^{2 + k} \cdot w_k \not \in \{-2, -1, 0, -1, 2\}$ as an integer. Informally, adding higher powers of two cannot change the sum's residue class modulo $2^{3 + k^*}$, which we have already seen equals $2^{2 + k^*}$. More formally,
\begin{equation*}\sum_{k = 0}^{j - 1} 2^{2 + k} \cdot w_k \equiv \sum_{k = 0}^{k^*} 2^{2 + k} \cdot w_k \equiv 2^{2 + k^*} \pmod{2^{3 + k^*}};\end{equation*}
this follows from the equality $\sum_{k = k^* + 1}^{j - 1} 2^{2 + k} \cdot w_k \equiv 0 \pmod{2^{3 + k^*}}$ (which actually holds regardless of the $w_k$). As the above equation's right-hand residue class modulo $2^{3 + k^*}$ is necessarily unequal to those of $-2$, $-1$, $0$, $1$, and $2$, the sum on its left-hand side is necessarily unequal to each of these values \textit{as an integer}.

Finally, in light of the above, ensuring that $\sum_{k = 0}^{j - 1} 2^{2 + k} \cdot w_k \not \in \{-2, -1, 0, -1, 2\}$ in $\mathbb{F}_q$ amounts simply to showing that $\sum_{k = 0}^{j - 1} 2^{2 + k} \cdot w_k \in \{-q + 3 \ldots , q - 3\}$ as an integer. In the worst case, this sum has an absolute value of $\sum_{k = 0}^{n - 1} 2^{2 + k} = 4 \cdot (2^n - 1)$; the statement of the first part of the theorem easily follows.

For the second part of the theorem, if $v_i \leq v_{1-i}$ then there exists exactly one $j \in \{0,\ldots,n\}$ such that $c_{i,j}=0$, and we have seen that this is the first $j$ for which for every $k < n$ it holds that $v_{0,k} \neq v_{1,k}$ (in case $v_i=v_{1-i}$, we have that $c_{i,n}=0$). The algorithm shuffles the vector $c_{i,0},\ldots,c_{i,n}$ and then multiplies each element with an independent random scalar. Thus, if $v_i \leq v_{1-i}$ exactly one element would be $0$, and all the rest are random in $\mathbb{F}_q \setminus \{0\}$. If $v_i > v_{i-1}$, then all elements in $c_{i,0},\ldots,c_{i,n}$, so after shuffling and multiplying with independent random scalars, the vector $d_{i,0},\ldots,d_{i,n}$ distributes uniformly in $\mathbb{F}_q^{n+1} \setminus \{0\}^{n+1}$. 
\end{proof}

\section{Proof of Theorem~\ref{thm:semi}}
\label{appx:thm:semi}

\begin{proof}
Regarding complexity, each party holds an input of length $n$ bits. It secret shares each one of the bits using a field of size $2^n$, and therefore it sends $O(n^2)$ bits. In addition, the parties run a secure coin-tossing protocol to toss $\lambda$ coins, and its complexity is $O(\lambda)$. In total we get $O(n^2 + \lambda)$.

We consider the protocol in an idealized (i.e., a hybrid) model where $\Pi_{\sf CT}$ is a secure coin-tossing protocol. Replacing it with a secure protocol follows by a composition theorem. Since Functionality~\ref{func:compare} is deterministic, we can separately show correctness and privacy, and we do not have to consider the joint distribution of the view of the adversary and the output of all parties. We start with showing correctness. 

\paragraph{Correctness.}
We first analyze the security of a modified protocol in which the two parties run $\pi_{\sf CT}$ and receive all the randomness they need for running Algorithm~\ref{alg:initial}. Note that this means that the parties do not use the pseudorandom generator at all throughout the execution. 

All operations in Algorithm~\ref{alg:initial} on the values $\left(v_{0,0},\ldots,v_{0,n-1}\right), \left(v_{1,0},\ldots,v_{1,n-1}\right)$ are linear, and therefore working on shares of those vectors, applying Algorithm~\ref{alg:initial} on each one of the shares separately, and then reconstructing, results with the same output as running Algorithm~\ref{alg:initial} directly on and the values $\left(v_{0,0},\ldots,v_{0,n-1}\right), \left(v_{1,0},\ldots,v_{1,n-1}\right)$. $P^*$ then runs Algorithm~\ref{alg:final} on the reconstructed values, and from Theorem~\ref{thm:correctness}, the result $(b_0,b_1)$ is $(v_0 \leq v_1,v_1\leq v_0)$, exactly as computed by the functionality. 

When replacing the random string with the result of $G(s)$, we get that the output is computationally-indistinguishable from the output of the functionality. 

\paragraph{Privacy.}
Here we separate to two different case. The first case is where $P_i$ is corrupted for some $i \in \bit$, and the second is when $P^*$ is corrupted.

\paragraph{\boldmath Case I: $P_i$ is corrupted.} The simulator receives the input ($v_i$) and the output ($b_0,b_1$) of the corrupted party, $P_i$. It has to generate the view of $P_i$. The view of $P_i$ consists of the following: 
\begin{MyEnumerate}
\item The internal random coins of $P_i$.
\item The result of the coin-tossing protocol, $s$. 
\item The random shares $\left( \share{v_{1 - i, j}}_i \right)_{j = 0}^{n - 1}$.
\item The output $b_i$ as sent by $P^*$. 
\end{MyEnumerate}
Clearly, all except for the last two items are independent random coins, which are easily simulatable. For the last item, the simulator receives them as input. Simulating the view of $P_i$ follows easily.

\paragraph{\boldmath Case II: $P^*$ is corrupted.}
The simulator receives the input and output of $P^*$, which consists of the two bits $(b_0,b_1)=(v_0 \leq v_1, v_1 \leq v_0)$. The view of $P^*$ consists of the shares 
$\left( \share{d_{0, j}}_i \right)_{j = 0}^n, \left( \share{d_{1, j}}_i \right)_{j = 0}^n$ for both $i=0$ and $i=1$.  $P^*$ in the protocol then reconstructs for every $j \in \{0,\ldots,n\}$ the values $d_{0, j}, d_{1,j}$ where 
\begin{eqnarray*}
\lefteqn{\left(d_{0, 0}, \ldots , d_{0, n - 1}\right), \left(d_{0, 0}, \ldots , d_{0, n - 1}\right)} \\
&& = \textsf{ComparisonInitial}\left( \left(v_{0, 0}, \ldots , v_{0, n - 1}\right),~\right.\\
&& \hspace{+25ex}\left. \left(v_{1, 0}, \ldots , v_{1, n - 1}\right)\right) \ .
\end{eqnarray*}
From Theorem~\ref{thm:correctness} it follows that all the values $d_{0, j}, d_{1,j}$ are uniform distributed in $\mathbb{F}_q\setminus \{0\}$ unless for one position, which reveals whether $v_0 \leq v_1$ and/or $v_1 \leq v_0$. In any case, simulating the values of $d_{0,0},\ldots,d_{0,n}$, $d_{1,0},\ldots,d_{n,0}$ is easy. The simulator then generates additive shares for those values, as obtained from $P_0$ and $P_1$ (note that the shared randomness between $P_0$ and $P_1$ is not in the view of the adversary).  

We remark that the view of $P^*$ also consists of the entire communication between $P_0$ and $P_1$, however, since it has no information on the private shared key between $P_0$ and $P_1$, all encrypted information between the two parties is indistinguishable from random (as follows from the secrecy of the encryption scheme), and therefore simulating this communication is straightforward. 
\end{proof}

\section{Malicious Security of Protocol~\ref{prot:malicious}}
\label{sec:malicious-proof}
\paragraph{Computing minimum on committed inputs.}
We slightly modify the protocol and the functionality, and this is the way we will use it to implement Functionalities~\ref{func:b2c} and~\ref{func:multi}. In particular, we assume that the server already holds commitments of the input values $v_0,v_1$, and the parties have the openings. Looking ahead, the clients first commit to their orders to the server; Later, when the server pairs two clients and we run the matching protocol, the bank provides the commitments to the counterparties. 
We have:
\begin{mdframed}
\begin{functionality}
[$\fmincom$ -- Computing minimum on committed inputs]
\label{func:fmincom}
Consider two players, $P_0$ and $P_1$, and a server $P^*$. 


\noindent
{\bf The functionality:}
\begin{enumerate}[nosep,leftmargin=*]
\item {\bf \boldmath In case of a corrupted $P_i$:} The honest $P^*$ inputs commitments $(V_0,V_1)$, and the honest $P_{1-i}$ sends $(v_{1-i},r_{1-i})$ such that $V_{1-i} = {\sf Com}(v_{1-i};r_{1-i})$. Send the corrupted $P_{i}$ the commitments $(V_0,V_1)$ and receive back $(v_i,r_i)$. Verify that $V_i = {\sf Com}(v_i; r_i)$, and otherwise abort. 
\item {\bf \boldmath In case of a corrupted $P^*$:} Receive from the each honest $P_j$ the input $(v_j,r_j)$, and  compute $V_j = {\sf Com}(v_j;r_j)$ for $j \in \{0,1\}$. Send $(V_0,V_1)$ to $P^*$. 
\item Compute $b_0 := (v_0 \leq v_1)$ and $b_1 := (v_1 \leq v_0)$. 
\item Send $b_0$ to $P_0$, $b_1$ to $P_1$ and $(b_0,b_1)$ to $P^*$. In addition, send to all parties $\min\{v_0,v_1\}$. 
\end{enumerate}

\end{functionality}
\end{mdframed}
To implement this functionality, we slightly modify Protocol~\ref{prot:malicious}: The parties do not send the commitments in the first step. Instead, $P^*$ sends both parties both commitments; In the last round, if a party outputs ${\sf true}$, then it sends a new commitment $V_i'$, and proves using ${\sf ComEq}$ that the two commitments are the same, and opens $V_i'$ to $P^*$. $P^*$ verifies that the commitments are the same, and sends this information also to $P_{i-1}$ as long as $b_0 \neq b_1$ (i.e., if they are both $1$, then both parties already know the minimum). Denote the modified protocol as $\pimincom$. 

\begin{theorem}
$\pimincom$ securely computes the $\fmincom$ (Functionality~\ref{func:fmincom}) in the presence of a malicious adversary who corrupts a single party. Each party has to send or receive $O(n(n+\lambda))$ bits where $n$ is the length of the input and $\lambda$ is the (computational) security parameter. 
\end{theorem}

\begin{proof}
We prove the security of the protocol in the hybrid model, where the underlying functionalities are 
${\cal F}_{\sf ZK}({\cal R})$ (see Section~\ref{appendix:sec:zk}). Using the composition theorem, we later conclude security in the plain model. 

%
We wrote the protocol in the plain model for brevity; We convert its description to a protocol in  the $\left({\cal F}_{\sf ZK}({\cal R}_{\sf ComEq}),{\cal F}_{\sf ZK}({\cal R}_{\sf BitProof}),{\cal F}_{\sf ZK}({\cal R}_{\sf OneMany}),{\cal F}_{\sf CT}\right)$-hybrid model and prove that the protocol is secure in this model. We discussed how to implement those functionalities in Section~\ref{appendix:sec:zk}), and using composition theorem we conclude that the protocol in the plain model is secure.

We need to prove that for any malicious adversary $\Adv$, the
view generated by the simulator $\Sim$ above is indistinguishable from the output in the real (hybrid)-model, namely:
$$\big \{\ideal_{\cF,\Sim(z),i}(u_0, u_1, \lambda)\big \}\approx_{comp}\big \{ \real_{\pi,\Adv(z),i}(u_0, x_1, \lambda)\big \}$$

We start with a simulator for the case where $P_i$ is corrupted, for $i \in \bit$. We later show a simulator for the case where $P^*$ is corrupted. Without loss of generality, assume that $i=0$. 

\paragraph{The simulator $\Sim$:}
\begin{MyEnumerate}
\item The simulator invokes the adversary $\Adv$ on an auxiliary input $z$
\item The simulator receives from the trusted party the commitment $V_0,V_1$. 
\item Simulate ${\cal F}_{\sf CT}$: Choose a random $s$ and give it to the adversary. If the adversary does not reply with ${\sf OK}$, then send $\bot$ to the trusted party.
\item Choose random shares for $\share{v_{1,j}}_0$ uniformly at random, and then compute commitments $V_{1,j,0}$. 
\item Choose random commitments $V_{1,j,1}$ such that $\prod_{j=0}^{n-1}(V_{1,j,0}\cdot V_{1,j,1})^{2^{n-1-j}} = V_1$. 
\item Send $\Adv$ the shares $\share{v_{1,j}}_0$, the full array $(V_{1,j,0},V_{1,j,1})_{j \in N}$, and open the commitments $V_{1,j,0}$ for $j \in N$. Moreover, simulate ${\cal F}_{\sf ZK}({\cal R}_{\sf ComEq}),{\cal F}_{\sf ZK}({\cal R}_{\sf BitProof})$ by sending ${\sf accept}$ to $\Adv$.  
\item Receive from $\Adv$ the full array $(V_{0,j,0},V_{0,j,1})_{j \in N}$ and the openings $\share{v_{1, j}}_1$. Check that the openings indeed open $V_{1, j, 1}$, for $j \in N$.
\item Simulate the ${\cal F}_{\sf ZK}({\cal R}_{\sf ComEq})$ invocation of Step~\ref{stp:malicious:comp-eq}: receive from $\Adv$ as part of its invocation the witness for $V_i$, which consists of $(v_i,r_i)$, and the witness $(v_i,r_i')$ of $\prod_{j = 0}^{n - 1} \left( V_{i, j, 0} \cdot V_{i, j, 1} \right)^{2^{n - 1 - j}}$. Verify that indeed $V_i = {\sf Com}(v_i;r_i)$ and that the ${\sf Com}(v_i,r_i')=\prod_{j = 0}^{n - 1} \left( V_{i, j, 0} \cdot V_{i, j, 1} \right)^{2^{n - 1 - j}}$. If these conditions hold, return ${\sf accept}$ to ${\sf Adv}$; otherwise, return ${\sf reject}$. 
\item Simulate the ${\cal F}_{\sf ZK}({\cal R}_{\sf BitProof})$ invocation of Step~\ref{stp:malicious:bit-proof} of each instance $V_{i,j,0}\cdot V_{i,j,1}$ for every $j \in N$; i.e., from each invocation also receive a witness and check that the witness satisfies the relation ${\cal R}_{\sf BitProof}$. For each check return ${\sf accept}$ or ${\sf reject}$. 
\item If all checks pass, send $v_i$ to the trusted party, where $v_i$ was received when simulating the ${\cal F}_{\sf ZK}({\cal R}_{\sf ComEq})$ functionality in the previous steps. Otherwise, send $\bot$. 
\item Receive $(b_0,b_1,v)$ from the trusted party as $P_i$'s output.
\item Run Algorithm~\ref{alg:initial} on the shares provided to $\Adv$ in the previous steps to obtain with randomness $s$ (which was the result of the ${\cal F}_{\sf CT}$) to obtain $\left( \share{d_{0, j}}_i \right)_{j = 0}^n, \left( \share{d_{1, j}}_i \right)_{j = 0}^n$ and the randomnesses $\left( s_{0, j, i} \right)_{j = 0}^n, \left( s_{1, j, i} \right)_{j = 0}^n$ and $\left(\left( D_{0, j, 1 - i} \right)_{j = 0}^n, \left( D_{1, j, 1 - i} \right)_{j = 0}^n\right)$. 
\item Receive from $\Adv$ the shares  $\left( \share{d_{0, j}}_i \right)_{j = 0}^n, \left( \share{d_{1, j}}_i \right)_{j = 0}^n$ and the randomnesses $\left( s_{0, j, i} \right)_{j = 0}^n, \left( s_{1, j, i} \right)_{j = 0}^n$ as in Step~\ref{stp:malicious-send-to-P-star}. Verify that those are the same values as computed. 
\item Compute $D_{i, j} := \mathsf{Com}(\share{d_{i, j}}_i; s_{i, j, i}) \cdot D_{i, j, 1 - i}$, for each $j \in N$, and then simulate invoking ${\cal F}_{\sf ZK}({\cal R}_{\sf OneMany})$ while sending $\Adv$ to output $1$. 
\item If $b_1 = 1$, then choose a random $V_1' = {\sf Com}(v, r')$. Send to $P_0$ the commitment $V_1'$, as well as ${\sf accept}$ for the ${\cal F}_{\sf ZK}({\cal R}_{\sf ComEq})$ on the instance $(V_1,V_1')$. 
\item If $b_0=1$, then receive from $\Adv$ commitment $V_0'$, and simulate an invocation of ${\cal F}_{\sf ZK}({\cal R}_{\sf ComEq})$ where expect to receive the opening of $V_0$ and $V_0'$. If indeed the openings are correct and $V_0'$ is a commitment of $v_0$, then send ${\sf OK}$ to the trusted party. Otherwise, send $\bot$. 
\item The simulator outputs the view of the adversary in the above execution. 
\end{MyEnumerate}

To show indistinguishability between the real and ideal execution, we consider the following experiments:
\begin{MyItemize}
\item ${\sf Exp}_1$: This is the real execution. The adversary $\Adv$ is run with the honest $P_{1}$ and $P^*$. The output of the experiment is the view of the adversary and the output of the honest parties $P_{1}$ and $P^*$. 
\item ${\sf Exp}_2$: We run the simulator $\Sim$ with the trusted party of Functionality~\ref{func:fmincom} with the following difference: When the honest party $P_{1}$ sends $v_{1}$ to the trusted party, the trusted party delivers $v_{1}$ to $\Sim$. Then, $\Sim$ uses $v_{1}$ as the input of $P_{1}$ to simulate the first steps of the protocol instead of choosing random commitments that match $V_1$ as in the simulation. The output of this experiment is the output of all honest parties as determined by the trusted party and the output of the simulator.  
\item ${\sf Exp}_3$: this is just as the ideal execution: We run the simulator $\Sim$ as prescribed above. The output of this experiment is the of all honest parties as determined by the trusted party and the output of the simulator. 
\end{MyItemize}

We claim that the output of ${\sf Exp}_1$ is computationally indistinguishable from ${\sf Exp}_2$. In particular, since the simulator receives the input of $P_{1}$, it essentially runs the honest $P_{1}$ while exactly simulating all ${\cal F}_{\sf ZK}$ functionalities for the different relations in the protocol. The output of the protocol in the real execution is determined from the shares that $P^*$ receives, which it then invokes Algorithm~\ref{alg:final}. In the simulation, the simulator uses the input $v_0$ as extracted when invoking the ${\cal F}_{\sf ZK}({\cal R}_{\sf CompEq})$ functionality, sending $v_0$ to the trusted party, which then computes $(b_0,b_1)=(v_0 \leq v_1, v_1\leq v_0)$. From the binding property of the commitment scheme, it is infeasible for the adversary to come up with a different $v_0$ that does not match $V_0$. Correctness of the protocol from similar reasoning as in the semi-honest case, and from linearity of Algorithm~\ref{alg:initial}.

 We then claim that the output of ${\sf Exp}_2$ is computationally indistinguishable from ${\sf Exp}_3$. The only difference is the view of the adversary are the commitments of $V_{1,j,1}$ that are never opened. Besides of those values, the entire view of the adversary is exactly the same (those are just ${\sf accept}$ messages received from the ${\cal F}_{\sf ZK}$ functionality). However, from the hiding property of the commitment scheme, the view of the adversary is indistinguishable. 
 
\paragraph{A corrupted $P^*$.}
In case where $P^*$ is corrupted, then the simulator simulates the entire communication between $P_0$ and $P_1$ as random strings.  $P^*$ has no input, and it receives from the trusted party the two bits $(b_0,b_1)$, the minimum $M$, and the commitments $V_0$ and $V_1$. From the bits $(b_0,b_1)$ it generates the two vectors $(d_{0,0},\ldots,d_{0,n})$ and $(d_{1,0},\ldots,d_{1,n})$ just as the simulator in the semi-honest case. It chooses random $s_{0,j},s_{1,j}$ and chooses random shares:
\begin{eqnarray*}
d_{0, j} := \share{d_{0, j}}_0 + \share{d_{0, j}}_1 \ , &~~~~~& s_{0, j} := s_{0, j, 0} + s_{0, j, 1} \\
d_{1, j} := \share{d_{1, j}}_0 + \share{d_{1, j}}_1 \ , &~~~~~& s_{1, j} := s_{1, j, 0} + s_{1, j, 1}
\end{eqnarray*}
It simulates $P_0$ sending  $\left( \share{d_{0, j}}_0 \right)_{j = 0}^n, \left( \share{d_{1, j}}_0 \right)_{j = 0}^n$ and the randomnesses $\left( s_{0, j, 0} \right)_{j = 0}^n, \left( s_{1, j, 0} \right)_{j = 0}^n$ to $\Adv$, and $P_1$ sending  $\left( \share{d_{0, j}}_1 \right)_{j = 0}^n, \left( \share{d_{1, j}}_1 \right)_{j = 0}^n$ and the randomnesses $\left( s_{0, j, 1} \right)_{j = 0}^n, \left( s_{1, j, 1} \right)_{j = 0}^n$ to $\Adv$. It then simulates the invocation of ${\cal F}_{\sf ZK}({\cal R}_{\sf OneMany})$: If $\Adv$ does not provide the correct statement, $\left( D_{i, j} \right)_{j = 0}^n$ for $D_{i, j}:= \mathsf{Com}(d_{i, j}; s_{i, j})$, and the correct witness (consisting of the index of the non-zero coordinate in $d$ and the opening of its commitment), then return ${\sf reject}$ to $\Adv$ and send $\bot$ to the trusted party. Otherwise, send ${\sf accept}$ to $\Adv$ and send ${\sf OK}$ to the trusted party. At the end of the invocation, if some party $P_j$ has to open the commitment to $\Adv$, then it sends a commitment $V_j' = {\sf Com}(M; r)$ with some random $r$, and simulates ${\cal F}_{\sf ZK}({\cal R}_{\sf ComEq})$ on instance $V_j,V_j'$ as ${\sf accept}$. 

It is easy to see that the simulator perfectly simulates the view of the adversary in the real execution. This follows from the properties of our comparison algorithm (Theorem~\ref{thm:correctness}), and that the simulator perfectly simulates the ${\cal F}_{\sf ZK}$ functionality. Finally, any wrong information $\Adv$ sends to ${\cal F}_{\sf ZK}$ results in ${\sf false}$ in the real execution. In contrast, in the ideal execution, the simulator sends $\bot$ to the trusted party which results in ${\bot}$ to both $P_0,P_1$ in the ideal execution. 
\end{proof}


\paragraph{Computing $\fctc$ functionality (Functionality~\ref{sec:func:client-to-client}.)}
To conclude client-to-client matching, the parties invoke $2|U|$ times ${\cal F}_{\sf min}$. For each possible symbol in $U$, the parties first commit to their short and long exposure for that symbol. They then invoke the committed minimum functionality twice: Once when $P_0$ inputs its long exposure and $P_1$ inputs short, and vice-versa. If a party is not interested in some symbol or in a particular side then it simply inputs $0$. The bank receives as output all the matches and can execute them directly. Security follows from just the composition theorem. We have:
\begin{corollary}
The above protocol securely computes Functionality~\ref{func:c2c} in the presence of one  malicious client or a malicious server. 
\end{corollary}

\section{Bank-to-Client Matching}\label{sec:twoparty}
In this section, we present a two-party variant of our main three-party comparison protocol which can be used for bank-to-client matching. We describe the functionality and then the protocol. 
\begin{mdframed}
\begin{functionality}[Two-party min, $\fminbtc$.]
\label{func:2-party-min}
{\bf Input:} The server $P_0$ holds $v_0$, and $P_1$ holds $v_1$, both in $\{0,\ldots,2^n-1\}$. 


\noindent
{\bf Output:} Both parties receive $\min\{v_0,v_1\}$.
\end{functionality}
\end{mdframed}

\begin{mdframed}
\begin{protocol}[Two-party protocol $\piminbtc$] \label{two-party}
%
\textbf{Input:} $P_0$ and $P_1$ hold integers $v_0$ and $v_1$, respectively, in $\{0, \ldots , 2^n - 1\}$. \\
\textbf{Setup phase:} A 
commitment scheme $(\mathsf{Com})$ and an encryption scheme $(\mathsf{Gen}, \mathsf{Enc}, \mathsf{Dec})$ are chosen. $P_0$ runs $(\textit{pk}, \textit{sk}) \gets \mathsf{Gen}(1^\lambda)$, and authenticates its public key $\textit{pk}$ with $P_{1}$. 

\medskip
\noindent
{\bf The protocol:} 
\begin{enumerate}[nosep,leftmargin=15pt]
\item {\bf \boldmath $P_0$ proceeds as follows: (Encryption):}
\begin{enumerate}[nosep,leftmargin=15pt]
\item Commit $V_0 \gets \mathsf{Com}(v_0; r_0)$, and send $V_0$ to $P_1$.
\item Compute the bit decomposition $v_0 = \sum_{j \in N} 2^{n - 1 - j} \cdot v_{0, j}$, for bits $v_{0, j} \in \{0, 1\}$. 
\item For each $j \in N$:\label{prot:generate-enc} compute additive homomorphic encryptions $A_{j} = \mathsf{Enc}_{\textit{pk}}({v_{0, j}})$. 
\item Send the full array $\left( A_{j} \right)_{j = 0}^{N-1}$ to $P_{1}$. \label{prot:malicious:send-array-A}
\item Compute: 
\begin{MyEnumerate}
\item $\pi \gets \mathsf{ComEq.Prove}\left( V_0, \prod_{j = 0}^{N - 1} \left( A_{j} \right)^{2^{j}} \right),$
\item $\pi_j \gets \mathsf{BitProof.Prove} \left( A_j\right)$, for all $j \in N$,
\end{MyEnumerate}
%
$P_0$ sends $\pi$, and $(\pi_{j})_{j = 0}^{N - 1}$ to $P_{1}$.
\end{enumerate}
\item {\bf \boldmath $P_1$ proceeds as follows: (Computing the minimum):}
\begin{enumerate}
\item[] Receive  $(\pi$,$(\pi_{j})_{j = 0}^{N - 1})$ and verify the following:
\item $\mathsf{ComEq.Verify} \left( \pi, V_{0}, \prod_{j = 0}^{N - 1} \left( A_{j} \right)^{2^{j}} \right)$,
\item $\mathsf{BitProof.Verify} \left(\pi_{j}, A_{ j} \right)$ for each $j \in N$.
If any of these checks fail, $P_1$ aborts.
%
%
\item Run Algorithm~\ref{alg:initial}, in parallel, on the ciphertexts
$A_j$ and on its own secret inputs. That is, $P_1$ runs:
\begin{eqnarray*}
\left( \left( {D_{0, j}} \right)_{j = 0}^N, \left( {D_{1, j}} \right)_{j = 0}^N \right) \gets\hspace{+10ex} \\\mathsf{ComparisonInitial}\left( (A_j)_{j=0}^{N-1}, (v_{1,j})_{j=0}^{N-1} \right)
\end{eqnarray*}
recall that the algorithm shuffles the two result vectors inside. 
\item $P_1$ sends $\left( \left( {D_{0, j}} \right)_{j = 0}^N, \left( {D_{1, j}} \right)_{j = 0}^N \right)$ to $P_0$.\label{stp:send-to-P0}
\end{enumerate}
\item {\bf \boldmath Party $P_0$ (Output reconstruction):} 
\begin{enumerate}[nosep,leftmargin=15pt]
\item Decrypt $d_{i,j} = \mathsf{dec}_{\textit{sk}}({D_{i, j}})$ for each $i=\{0,1\}$ and $j=\{0,\ldots,n-1\}$
\item Execute Algorithm~\ref{alg:final}, that is: \vspace{-1ex}
$$
(b_0, b_1) := \mathsf{ComparisonFinal} \left( \left( d_{0, j} \right)_{j = 0}^N, \left( d_{1, j} \right)_{j = 0}^N \right) \ . 
$$ 
\item Write $u$ for the index such that $b_u$ is true and sets $u:=0$ if both are. 
\item Compute $\pi' \gets \textsf{OneMany.Prove}\left( \left( D_{u, j} \right)_{j = 0}^{n-1} \right)$, and sends $\pi'$ and $u$ to $P_1$. If $u=1$ then send also $v_1$. 
\end{enumerate}
\item[]
\item {\bf \boldmath Party $P_1$ (output reconstruction):} 
\begin{enumerate}[nosep,leftmargin=15pt]
\item $P_1$ verifies $\textsf{OneMany.Verify}\left( \pi', \left( D_{u, j} \right)_{j = 0}^{N-1} \right)$. If verification passes, then $P_1$ outputs $v_u$.
\end{enumerate}
\end{enumerate}
\end{protocol}
\end{mdframed}

\begin{theorem}
Protocol~\ref{two-party} securely computes Functionality~\ref{func:2-party-min} in the presence of a malicious $P_0$ or a  semi-honest $P_1$, assuming secure commitments and zero-knowledge functionalities. 
\end{theorem}
\begin{proof}
Since the protocol is not symmetric, we separate between corrupted $P_0$ and corrupted $P_1$. 

\paragraph{Security against a malicious $P_0$:} We show security assuming an ideal commitment and a zero-knowledge scheme. The simulator plays the role of the ideal functionalities for those sub-protocols and receives in particular the input $v_0$ and the decomposition bits when $P_0$ calls to the zero-knowledge and commitment functionalities. It verifies that $P_0$ provided consistent input. If not, the simulator aborts in the corresponding functionality. 

If the inputs are consistent, then given the input $v_0$ it sends it to the ideal functionality and receives back the minimum and whether the two integers are the same. It computes from that $(b_0,b_1)$, and simulates values $d_{0,j},d_{1,j}$ uniformly at random as per Theorem~\ref{thm:correctness}, and sends to $P_1$ encryption of those values under $P_0$'s public key. It then verifies that the one-out-of-many proof of $P_0$ is correct, and if so it sends ${\sf OK}$ to the trusted party, in which it delivers to $P_1$ its output in the ideal world. It is easy to see that the simulator perfectly simulates the view of the adversary in the ideal world.

\paragraph{Security against a semi-honest $P_1$:} In that case, the simulator sends to the adversary a commitment to $0$, and bit decomposition as all zeros. Moreover, it simulates the zero-knowledge proofs as succeeding. It then receives encryptions of $(d_{0,j},d_{1,j})$. It receives from the trusted party the minimum value and simulates the final ZK proof as successful. 
\end{proof}

\paragraph{Concrete instantiation.}
We implement Protocol~\ref{two-party} using the ElGammal encryption scheme, where the message $m$ is part of the exponent (thereby we get additive homomorphism). I.e., for a public key $h \in \mathbb{G}$, ${\sf Enc}_pk(m) = (g^r,h^rg^m)$. Note that $P_1$ does not have to decrypt the ciphertexts, but just identify an encryption of 0. Given a secret key $x$ such that $h=g^x$ and a ciphertext $c=(c_1,c_2)$, this is done by simply comparing $c_2/c_1^x$ to $g^{0}$.

\paragraph{Implementing Functionality~\ref{func:b2c}.} 
To implement the bank-to-client functionality, the parties invoke $2|U|$ times the two-party minimum functionality. For each possible symbol in $U$, the parties invoke the minimum functionality where once the bank inputs its interest in symbols to buy and the client in sell, and one time where the bank inputs whether it wishes to sell and the client inputs whether it wishes to buy. If a party is not interested in some symbol or in a particular side, it simply inputs $0$. The results of all minimum invocations will reveal the matches. Security follows from composition. We have:
\begin{corollary}
The above protocol securely computes Functionality~\ref{func:b2c} ($\fbtc$) in the presence of a semi-honest bank or a malicious client. 
\end{corollary}

\section{Realizing Functionality~\ref{func:multi}}\label{sec:mainprotocol}

In this Section we present protocol~\ref{prot:multi} for the functionality~\ref{func:multi} in the $\fmincom$-hybrid (Functionality~\ref{func:fmincom}) model.  When matching
begins, orders pertaining to the same security and of opposite direction are matched. After each particular
match, the server decrements both orders' quantities by the matched amount, and dequeues whichever among the two orders is empty (necessarily at least one will be). The commitments' homomorphic property allows the engine to appropriately decrement registrations between matches. 



\begin{mdframed}
\begin{protocol}[$\pimulti$---multiparty matching] \label{prot:multi}
\leavevmode

\noindent
Upon initialization, $P^*$ initializes a list $P=\emptyset$ and two vectors ${\cal L}$ and ${\cal S}$ of size $n$, where $n$ bounds the number of possible clients. 

\noindent
{\boldmath $\pimulti.{\sf Register}(P_i)$: }
When the command is invoked, then $P_i$ sends to $P^*$ commitments $\widetilde{L_i} = {\sf Com}(L_i;r_i)$ and $\widetilde{S_i} = {\sf Com}(S_i; s_i)$ for some random $r_i, s_i$. $P^*$ set ${\cal L}[i] = \widetilde{L_i}$ and ${\cal S}[i] = \widetilde{S_i}$. The party $P_i$ locally stores $(L_i,r_i)$ and $(S_i,s_i)$. $P^*$ also adds $i$ to $P$.

\smallskip
\noindent
{\boldmath $\pimulti.{\sf Process}()$:}
\begin{enumerate}[nosep,leftmargin=*]
\item $P^*$ chooses a random ordering $O$ over all pairs of $P$. 
\item For the next pair $(i,j) \in O$, try to match between $P_i$ and $P_j$:
\begin{enumerate}[nosep,leftmargin=*]
\item Invoke Functionality~\ref{func:fmincom} where $P_i$ inputs $(L_i, r_i)$ and $P_j$ inputs $(S_j, s_j)$, $P^*$ inputs $({\cal L}[i], {\cal S}[j])$. Let $M_0$ be the resulting minimum. $P_i$ sets $L_i = L_i-M_0$ and $P_j$ sets $S_j = S_j - M_0$. If $M_0 \neq 0$, then execute the match. $P^*$ homomorphically evaluates the commitments $({\cal L}[i], {\cal S}[j])$ to subtract $M_0$. 
\item Invoke Functionality~\ref{func:fmincom} where $P_i$ inputs $(S_i, s_i)$ and $P_j$ inputs $(L_j, r_j)$, $P^*$ inputs $({\cal S}[i], {\cal L}[j])$. Let $M_1$ be the resulting minimum. $P_i$ sets $S_i = S_i-M_1$ and $P_j$ sets $L_j = L_j - M_1$. If $M_1 \neq 0$, then execute the match. $P^*$ homomorphically evaluates the commitments $({\cal L}[i], {\cal S}[j])$ to subtract $M_1$. 
\end{enumerate}
\end{enumerate}
\end{protocol}
\end{mdframed}

\begin{theorem}
Protocol~\ref{prot:multi} securely implements Functionality~\ref{func:multi} in the $\fmincom$-hybrid model (Functionality~\ref{func:fmincom}) in the presence of a malicious client or a semi-honest server. 
\end{theorem}
\begin{proof}
We first simulate the case where some client is malicious:
\begin{MyEnumerate}
\item Invoke the adversary $\Adv$. Whenever the adversary sends $\tilde{L_i}, \tilde{S_i}$, store $\widetilde{L_i}$ and $\widetilde{S_i}$ in ${\cal L}[i]$ and ${\cal S}[i]$, respectively. 
\item Whenever received $(i,j,M_0,M_1,b_i)$ from the trusted party, choose random commitments ${\cal L}[j]$ and ${\cal S}[j]$ and simulate two invocations of $\fmincom$:
\begin{MyEnumerate}
\item Send to $P_i$ the commitment ${\cal L}[i]$ and a random commitment ${\cal S}[j]$. Receive back $b_i$ and $M_0$. Homomorphically subtract $M_0$ from ${\cal L}[i]$. 
\item Send to $P_i$ the commitment ${\cal S}[i]$ and a random commitment ${\cal L}[j]$. Receive back $b_j$ and $M_1$. Homomorphically subtract $M_1$ from ${\cal L}[i]$. 
\end{MyEnumerate}
\end{MyEnumerate}
From inspection, it is easy to see that the view of $P_i$ is the same in the simulation and in the real world since the simulator perfectly simulates the invocations of Functionality~\ref{func:fmincom} from the information it receives from the trusted party. Moreover, the underlying committed values of $P^*$ in the real execution are the same as in the ideal process for the same ordering $O$ (which have the exact same distribution in the real and ideal). This can be shown by induction over the different invocations of the ordering of $O$: in each execution, the functionality~\ref{func:fmincom} gives $P^*$ the same as in the ideal execution of $\fmult$. From the homomorphic evaluation of the commitment scheme, the subtraction of the values is equivalent to $\fmult$ subtracting the minimum values.

\paragraph{The case of a corrupted $P^*$. }
We show security in the presence of a \emph{semi-honest} $P^*$. 
First, whenever received a message ${\sf registered}(P_i)$ from the trusted party, simulate $P_i$ sending to $P^*$ some random commitments $\widetilde{L_i}, \widetilde{S_i}$, and store those commitments. 

The simulator then repeatedly receives from the trusted party values $(i,j,M_0,M_1,b_0,b_1)$ and has to simulate the view of $P^*$ in the two invocations of $\fmincom$. First, from all those values, it can extract random coins that will result in the same random ordering $O$ as the one that the trusted party chose. Moreover, for each such $(i,j,M_0,M_1)$, $(b_0^0,b_0^1,b_1^0,b_1^1)$, the simulator pretend sending the adversary $\widetilde{L_i}$, $\widetilde{S_j}$ as coming from $\fmincom$; and then the output $(b_0^0,b_1^0)$ and $M_0$. For the second invocation, pretend sending the adversary $\widetilde{S_i}$, $\widetilde{L_j}$ and the output $(b_1^0,b_1^1)$ and $M_1$. Then, update $\widetilde{L_i}$, $\widetilde{S_i}$ and $\widetilde{L_j}$, $\widetilde{S_j}$ as in the protocol. From inspection, it is clear that the view of the adversary is the same in both executions.

\end{proof}

\section{Range Bank-to-Client Functionality}
\label{sec:range}

In this section, we consider the setting in which the client submits a minimum and a maximum amount per stock instead of a single amount. The current system is set to accept a minimum and a maximum quantity from the clients to be matched against the bank's inventory. As a first step, the bank runs Functionality~\ref{func:b2c} on every client one by one only on their minimum quantity until the bank exhausts its inventory. For the case where the bank did not exhaust its inventory after processing all the clients, it runs again Functionality~\ref{func:b2c} on every client one by one on their maximum quantity. This process is presented in Functionality~\ref{func:b2cv2}. The requirement on the minimum and maximum quantity is imposed by the business for this use case as a greedy approach in an effort to satisfy more clients with the bank's inventory. This is not a necessary requirement, but the bank views it as an additional feature to satisfy more clients. This greedy approach is not ideal since, for instance, the minimum quantity of the first client can exhaust the full inventory of the bank. We leave it as an open problem to devise a secure optimization algorithm for better allocations across all clients. 

This additional range feature is considered only when we seek for matches between the bank and a client in an effort to maximize the number of matches the bank can accommodate from its own inventory. Thart said, we do not explicitly consider it for the client-to-client matching. 


\begin{mdframed}

\begin{functionality}[$\mathcal{F}_{\text{B2C}}$--Range Bank-to-client functionality]
\label{func:b2cv2}
The functionality is parameterized by the set of all possible securities to trade, a set $U$. 

\noindent
{\bf Input:} The bank $P^*$ inputs lists of orders $({\sf symb}_i^*, {\sf side}_i^*, {\sf amount}_i^*)$ where ${\sf symb}_i \subseteq U$ is the security, ${\sf side}_i^* \in \{{\sf buy}, {\tt sell}\}$ and ${\sf amount}_i^*$ is an integer. A client $P_i$ sends its list of the same format but with a minimum and a maximum amount, $({\sf symb}_i^C, {\sf side}_i^C, {\sf MinAmount}_i^C,{\sf MaxAmount}_i^C)$.

\noindent
{\bf Output:} Initialize a list of ${\sf Matches}$ $M$. Choose a random order $O$ over all clients.  
\begin{itemize}
    \item For the next client  $j$ in $O$ 
    and for every $i$, 
    such that ${\sf symb}_i^{*} = {\sf symb}_j^{C}$ and ${\sf side}_i^{*} \neq {\sf side}_j^{C}$, add $({\sf symb}_i^{*}, {\sf side}_i^{*}, {\sf side}_j^{C}, v=\min\{{\sf amount}_i^{*}, {\sf MinAmount}_j^{C}\})$ to $M$ and update ${\sf amount}_i^{*} = {\sf amount}_i^{*}-v$. 
\end{itemize}
Then, check maximum amounts for the $j$ in $M$:
\begin{itemize}
    \item For the next $j$ in $O$ and every $i$ such that ${\sf symb}_i^{*} = {\sf symb}_j^{C}$ and ${\sf side}_i^{*} \neq {\sf side}_j^{C}$, add $({\sf symb}_i^{*}, {\sf side}_i^{*}, {\sf side}_j^{C}, v=\min\{{\sf amount}_i^{*}, {\sf MaxAmount}_j^{C}\})$ to $M$ and update ${\sf amount}_i^{*}={\sf amount}_i^{*}-v$. 
\end{itemize}

\end{functionality}
\end{mdframed}

Implementing this protocol is quite straightforward, given the protocols we already have. The bank simply runs the minimum functionality against (random) ordering of the clients with their minimum amount, and then with their maximum amount, while updating the amount of each symbol along the way.

\section{Prime Match System Performance Cont.} 
  In Figure~\ref{fig:finalmatch}, we present a sample of the user interface of the Prime Match system. 
 
      \begin{figure*}
       \centering
        \includegraphics[scale=0.25]{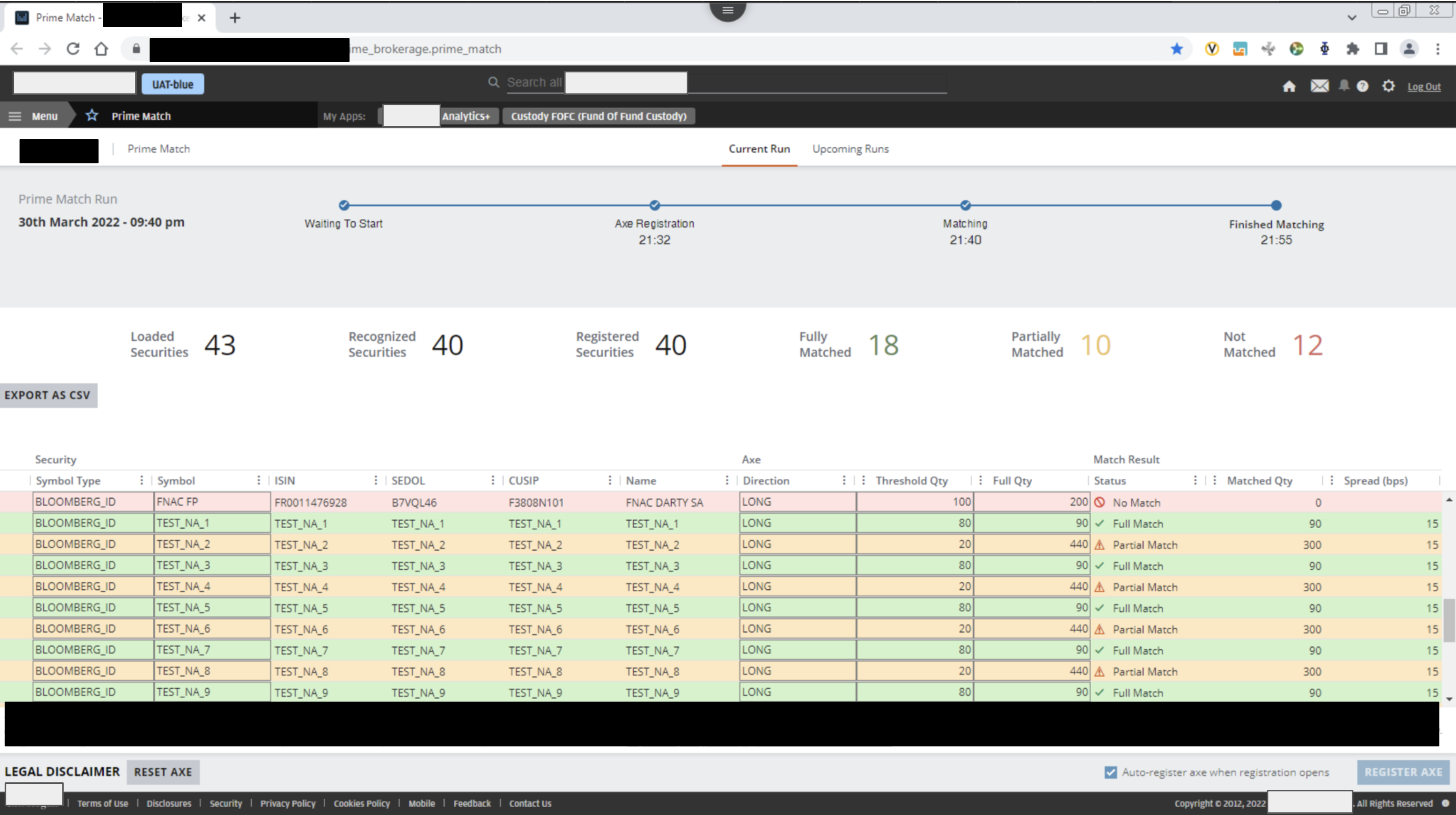}
        \caption{Prime Match User Interface of a completed matching run.}
       \label{fig:finalmatch}
    \end{figure*} 

 \begin{table*}[ht]
  \begin{center}
  \begin{footnotesize}
    \begin{tabular}[t]{c|c|c|c|c|c|cc|c}
      \toprule
         & &  Throughput & Registration  & Registration  & Matching  & Matching \\
       Number of    & &   (transactions/& received msg      & sent msg    & received msg  & sent msg \\

     Symbols  &Latency (sec)  &    sec)  &  Size (MB)    & Size (MB) &Size (MB) &Size (MB) \\
      \midrule
     100  &23.85 ($\pm$9.19 ) &8.38&0.062 &0.036 &16.225&4.510\\

      200       &29.49 ($\pm$2.85 ) &13.56&0.124 &0.073 &32.451&9.021\\
  


      500  &65.04 ($\pm$1.41 ) &15.37&0.310 &0.184 &81.128&22.552\\

      1000       &136.93 ($\pm$3.32 ) &14.60&0.621 &0.370 &162.257&45.105\\

      2000 &244.37 ($\pm$10.31 )  &16.36&1.243   &0.749   &324.507&90.210\\
  
      3000      &428.48 ($\pm$32.19 )   &14.00&1.865   &1.128   &486.739&135.260\\
 
      4000  &1077.60 ($\pm$27.46 ) &7.42&2.487 &1.507 &647.758&180.043\\

      5000       &1314.44 ($\pm$39.15 ) &7.60&3.109 &1.886 &809.191&224.641\\
  
      6000 &1558.58 ($\pm$46.96 )  &7.69&3.731   &2.265  &971.875&269.997\\
      \bottomrule
    \end{tabular}
  \end{footnotesize}
  \end{center}
  \caption{Performance of Bank-to-Client Prime Match for two clients. The number of symbols processed in the production code is double since clients provide both a minimum and a maximum quantity to be matched. This is reflected at the Throughput column which is multiplied by a factor of two. }
  \label{fig:primematchperf}
\end{table*}%

\section{Comparison to Generic MPC}
\label{sec:generic-mpc}
In this section, we compare the performance of our protocol to generic MPC protocols. Concretely, we look at the client-to-client matching. Towards that end, we estimate the size of the circuit that we have to use to implement such a functionality. 

Our main client-to-client matching protocol has two phases, in which the client first commits to their orders, and then the parties run in a second phase that the orders they provide match the committed ones. Implementing that using generic MPC protocols is relatively expensive as they must involve cryptographic primitives. Implementing group operations in a circuit, which involves exponentiation and public-key primitives is clearly prohibitively expensive, so we would like to use symmetric-key primitives. 

Specifically, we look at a commitment scheme in which the server (the bank) chooses some random seed $s$ and publishes it to all parties. To commit to a value $v$, a party chooses a random $r$ and computes $H(s,v,r)$ where $H$ is a CRHF. This commitment scheme is secure when modeling $H$ as a random oracle. When instantiating $H$ using SHA256, this (heuristically) gives 128-bit security for the binding property (a collision implies breaking the commitment). 

A Boolean circuit for SHA256 has 22K number of AND gates, 
and depth 1607 (see~e.g.,\cite{CampanelliGGN17,Sha256circuit}). This means that even ignoring the costs for comparisons themselves, we have:
\begin{MyItemize}
    \item Using any garbling-like protocol (say, variants of BMR) and with 128-bit keys for the labels, a circuit that checks that the commitments are correct will require communicating $22K \cdot 128 \cdot 2 \cdot 2$ bits ($128$ bit keys; using half gate optimization, so two ciphertexts per gate; for each symbol, we need at least two commitments - one from each one of the clients). The communication would be at least 722KB for processing a symbol. This is compared to roughly 25 KB per symbol by our protocol, as reflected by Table~\ref{fig:3partyAWS}. This is at least $\times 30$ more expensive than our solution, while the actual performance is likely to be much more expensive. We ignored here the cost of the comparisons themselves, or if the  clients have to be paired later again then the circuit has to compute the modified commitment after reducing the minimum (which requires further evaluations of SHA inside the circuit).
    \item Using any GMW-like protocol that requires number of rounds that is proportional to the depth of the circuit, the number of rounds would be at least 1607. This is compared to 3 rounds as in our protocols. This is an increase in a factor of $\times 500$. 
\end{MyItemize}

\section{Challenges in Implementation}\label{app:imp}

     \begin{figure}
        \centering
        \includegraphics[scale=0.35]{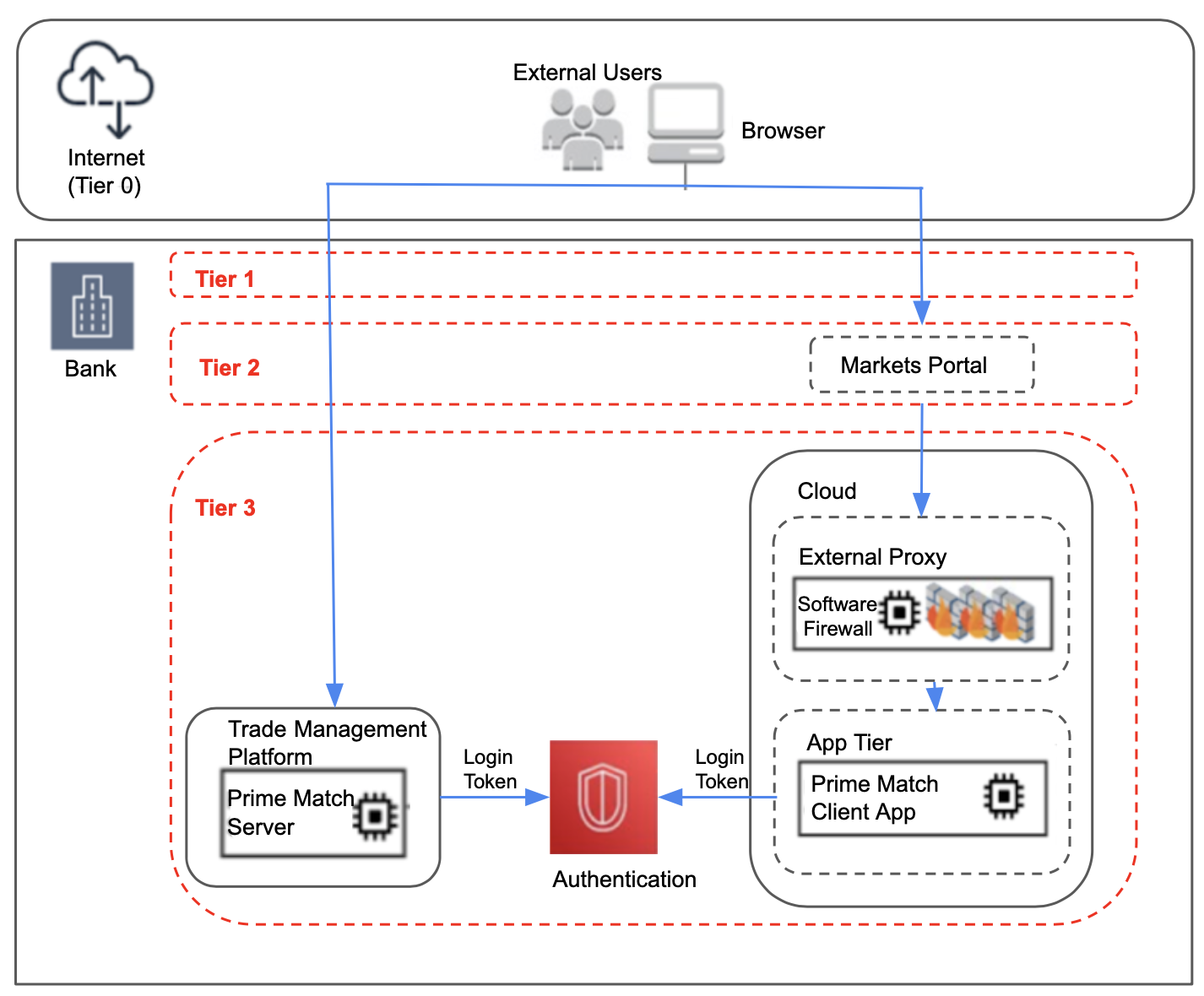}
        \caption{Prime Match system architecture connecting external users/clients to the bank's network.
        }
        \label{fig:arch}
    \end{figure}

          \begin{figure*}
       \centering
        \includegraphics[scale=0.48]{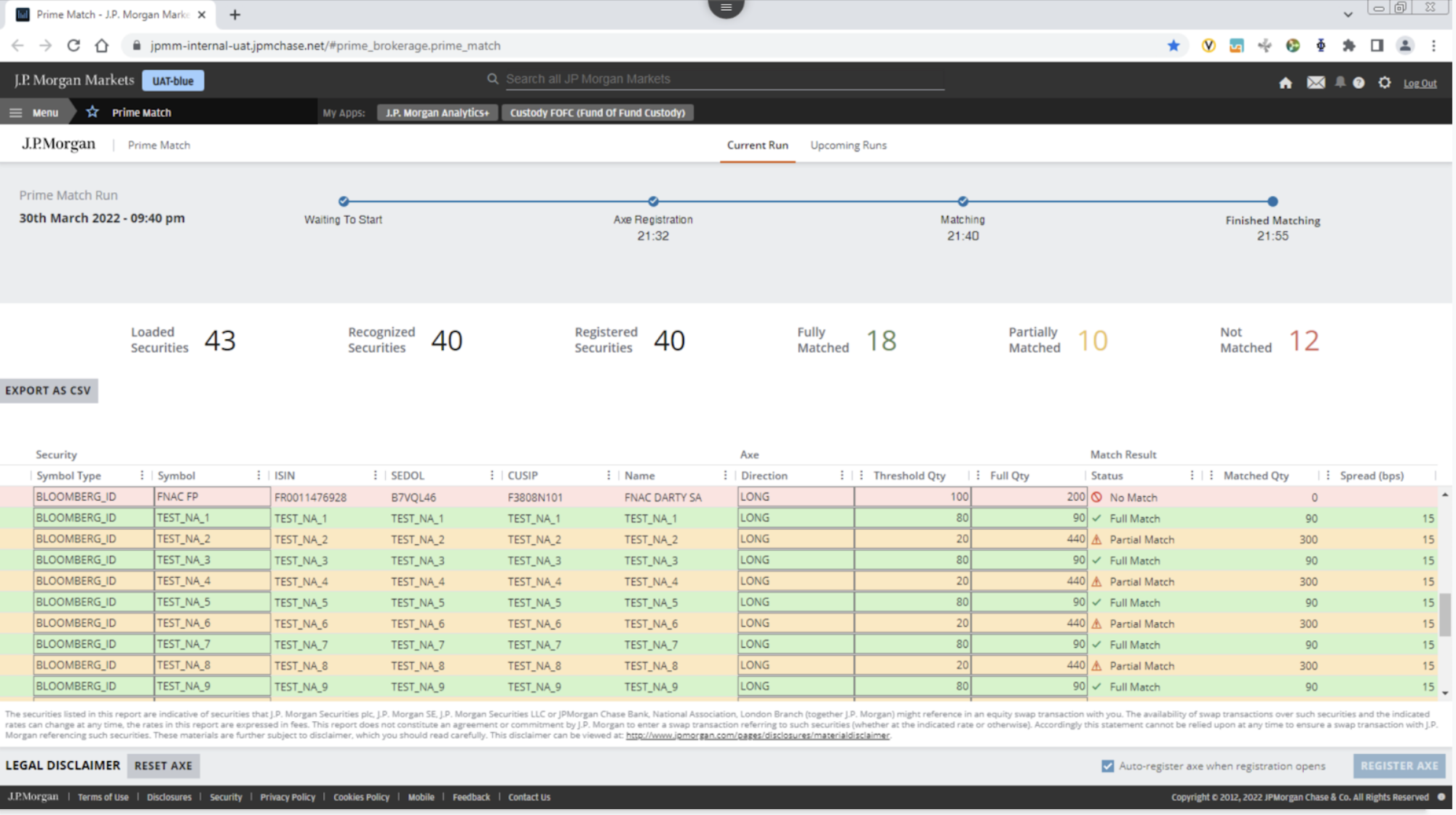}
        \caption{Prime Match User Interface of a completed matching run.}
       \label{fig:finalmatch}
    \end{figure*} 

Privacy-preserving auctions have been the holy grail of practical MPC since~\cite{BogetoftCDGJKNNNPST09}. However, no financial institutions took a step forward to use such tools. To begin with, we got connected to the bank's business department of quantitative research to test the appetite for a form of privacy-preserving auctions. Thus, the Axe inventory use case was chosen. 

\paragraph{Pre-Production Challenges.}
 Releasing a secure multiparty computation (MPC) product in a big organization that had no other products utilizing these techniques was itself challenging. Furthermore, there  were no other products across the street (market) to showcase the feasibility of such a solution. 

 As a first step, a proof of concept (POC) was implemented to show its feasibility. 
 
 Given the innovative nature of the product, the green light for production was given after a long process. More specifically, the organization decided to test the appetite for the possible product of its valuable clients (hedge funds). Almost 20 hedge funds were visited in 6 months to demo the POC and to hear their highly appreciated feedback. Given the client feedback, the organization decided to move another step forward and ask the clients under what conditions they will utilize such a product. In particular, the clients provided a set of different features and requirements they would like to see in the product. Notable requirements were:
 \begin{MyItemize}
     \item No communication with other clients, only communication with the bank.
     \item No resources for preprocessing data.
     \item No installation of code on client's machines - a web-based application was required.
     \item Stronger security guarantees than semi-honest security.
     \item Peer review of the solution.
 \end{MyItemize}

 The POC was updated to satisfy the above requirements. Another round of demos to clients was conducted. After several months the business gave the green light to build the product and allocated resources to enable it. The team consisted of cryptographers, quants, tech quants, expert programmers on different topics (such as secure multiparty computation, UI, web-based development etc.), cybersecurity experts, product managers, legal, and stakeholders.  

\paragraph{Production Challenges - Lessons Learned}
During the development of Prime Match, we faced several
critical technical challenges. Those challenges were mainly
due to the complex design of the organization's infrastructure for security measures.

\begin{MyItemize}
    \item Coding in the Trade Management platform and the Markets Portal, see Figure ~\ref{fig:arch}, needed to follow certain practices (which easy maintenance too) which were not followed in the POC stage. POC was reprogrammed to follow current practices. 
    \item Only a limited set of open-source libraries are allowed in the Trade Management platform. For example, importing the open-source library emscripten docker image to the Trade Management platform. This process went through rigorous request and approval processes to show that this library does not contain any code that would secretly send data to the internet.
    \item Integration to the markets portal which is an existing service for external clients, in order to access some of the organization's trading applications, was not easy since no other application required interaction between the server and the clients. Prior apps involve simple downloads of data from clients. 
    \item  Set up a pathway that allows network traffic to flow from the clients to the Prime Match server with the use of web socket protocol. In particular, setting up Psaas service in Tier 1 and 2, which stands for Perimeter Security as a Service. Extensive review and approval process were required.
    \item Placing a physical server and setting up an Apache WS tunnel in tier 1 to enable external clients to transact and transact within tight latency requirements. In particular, we had to implement the entire WS tunnel proxy from scratch since the Web Sockets protocol was not used within these platforms before.
    \item Run several internal tests ranging from UI requirements to cyber security requirements (such as timing attacks). 
    \item Extensive code review from internal and external experts. 
    \item Integrate the UI component with the markets portal which required extensive approvals. 
    \item Integration of our code with WebAssembly to achieve near-native performance. To this end, we also had to handle garbage collection. See more details on the chosen programming languages in Section~\ref{sec:prime-match-perf}.

\end{MyItemize}

Overall each step required extensive review and approval processes. An important aspect that was not listed above is the last process of extensive testing, verification of the code, and beta testing the product with the help of a limited number of clients.

\end{document}